\documentclass[letterpaper,12pt]{article}
\usepackage[margin=1.00in]{geometry}

\usepackage{endnotes, amsmath, amssymb, amsfonts, amsthm, graphicx, color}
\usepackage{comment, url, array, amscd}
\usepackage{setspace}
\DeclareMathOperator*{\essinf}{ess\,inf}

\newcommand{\E}           {\mathbb{E}}        
\newcommand{\bF}          {\mathbb{F}}        
\newcommand{\bI}          {\mathbb{I}}        

\newcommand{\bN}          {\mathbb{N}}        
\newcommand{\bP}          {\mathbb{P}}        
\newcommand{\bQ}          {\mathbb{Q}}        
\newcommand{\R}           {\mathbb{R}}        
\newcommand{\x}           {\textbf{x}}        
\newcommand{\y}           {\textbf{y}}        

\newcommand{\sB}          {\mathcal{B}}       
\newcommand{\sC}          {\mathcal{C}}       
\newcommand{\sD}          {\mathcal{D}}       
\newcommand{\sE}          {\mathcal{E}}       
\newcommand{\sF}          {\mathcal{F}}       
       
\newcommand{\sH}          {\mathcal{H}}       

\newcommand{\sL}          {\mathcal{L}}
\newcommand{\sM}          {\mathcal{M}}

\newcommand{\sR}          {\mathcal{R}}

\newcommand{\eps}        {\varepsilon}
\newcommand{\phimin}     {\phi_{\text{min}}}

\newcommand{\Xhat}       {\hat{X}}
\newcommand{\xhat}       {\hat{x}}
\newcommand{\Leb}        {\text{Leb}}

\newcommand{\Qmin}       {{\mathbb{Q}^{\infty}}}  
\newcommand{\bmo}        {{\text{BMO}}}      
\newcommand{\adm}        {\sH_{\text{adm}}}  
\newcommand{\LargerCdot} {\raisebox{-0.25ex}{\scalebox{1.4}{$\cdot$}}}

\numberwithin{equation}{section}
\theoremstyle{plain}                
\newtheorem{theorem}{Theorem}[section]
\newtheorem{lemma}[theorem]{Lemma}
\newtheorem{proposition}[theorem]{Proposition}
\newtheorem{corollary}[theorem]{Corollary}
\theoremstyle{definition}           
\newtheorem{definition}[theorem]{Definition}
\newtheorem{example}[theorem]{Example}

\newtheorem{assumption}[theorem]{Assumption}
\theoremstyle{remark}               
\newtheorem{remark}{Remark}[section]

\begin{document}

\begin{center}
  \vskip 1cm
  \Large{\textbf{Stability of Utility Maximization in Nonequivalent Markets}\footnote{The author would like to thank Dmitry Kramkov, Steve Shreve, and the anonymous reviewers for their constructive comments and suggestions.}}
  \vskip .6cm
  \large{
    Kim Weston\footnote{Department of Mathematical Sciences, Carnegie Mellon University, Pittsburgh, PA  15213\\ e-mail: kimberly@andrew.cmu.edu}\\
    }
  \vskip .2cm
  \today  
\end{center}

\begin{abstract}
  Stability of the utility maximization problem with random endowment and indifference prices is studied for a sequence of financial markets in an incomplete Brownian setting.  Our novelty lies in the \textit{nonequivalence} of markets, in which the volatility of asset prices (as well as the drift) varies. 
  Degeneracies arise from the presence of nonequivalence.  In the positive real line utility framework, a counterexample is presented showing that the expected utility maximization problem can be unstable.  A positive stability result is proven for utility functions on the entire real line.
\end{abstract}

\section{Introduction}\label{section:intro}
As part of Hadamard's well-posedness criteria, stability of the utility maximization problem with random endowment is studied with respect to perturbations in both volatility and drift.  Specifically, we seek to answer the question:
\begin{center}
\textit{
What conditions on the utility function and modes of convergence on the sequence of volatilities and drifts guarantee convergence of the corresponding value functions and indifference prices?
}
\end{center}

Perhaps surprisingly, convergence can fail even in the tamest of settings when the utility function is finite only on $\R_+$ and the volatility can vary.  We present a simple counterexample in a stochastic volatility setting with power utility.  When the utility function is finite only on $\R_+$, the admissibility criterion is harsh: negative values in terminal wealth plus random endowment equate to minus infinity in utility.  When volatility can vary, a contingent claim that is replicable only in the limiting market requires strictly more initial capital in every pre-limiting market 
in order to avoid a minus infinity contribution towards expected utility.  As part of the counterexample, we prove a positive convergence result in which the limiting market adopts an additional admissibility condition that is implicitly present in each pre-limiting market.

When the investor's utility function is finite on the entire real line, the admissibility criterion is different.  Our main result provides conditions on the utility function and on the sequence of markets so that we have convergence of the value functions and indifference prices.  We consider a similar setup to \cite{LZ07SPA}, and our main assumptions are analogous to theirs.  The only non-standard assumption we require is an assumption on the limiting market.  The significant difficulty stems from the growth of the dual utility function at infinity because in contrast to utilities on $\R_+$, the conjugate of a real line utility grows strictly faster than linearly at infinity.  We provide two sufficient conditions.  These conditions include:
\begin{enumerate}
\item The first condition applies to a contingent claim that is replicable in the limiting market yet not replicable in any pre-limiting market.  The corresponding stability problem is relevant when a claim's underlying asset is not liquidly traded but is closely linked to a liquidly traded asset.  This situation arises, e.g., when hedging weather derivatives by trading in related energy futures or when an executive wants to hedge his position in company stock options but is legally restricted from liquidly trading his own company's stock.  Practical and computational aspects of this problem are considered by \cite{D00preprint}, \cite{M04QF}, and in more generality by \cite{FS08AAP}.

\item The second sufficient condition requires exponential preferences and additional regularity of the limiting market but places no restrictions on the claim's replicability.  This case covers a general incomplete Brownian market structure under a mild $\bmo$ condition on the limiting market.  The connection between $\bmo$ and exponential utility is long established; see, for example, \cite{6AP} and \cite{GR02AP}.
\end{enumerate}

The questions of existence and uniqueness for the optimal investment problem from terminal wealth are thoroughly studied.  The surrounding literature is vast, and only a small subset of work is mentioned here.  For general utility functions on $\R_+$ in a general semimartingale framework, \cite{KS99AAP} finish a long line of research on incomplete markets without random endowment.  In \cite{CSW01FS}, this work is extended to include bounded random endowment, while \cite{HK04AAP} study the unbounded 
random endowment case. 
For utility functions on $\R$ in a locally bounded semimartingale framework, \cite{S01AAP} studies the case with no random endowment, while \cite{OZ09MF} handle the unbounded random endowment case.  In \cite{BF08AAP}, the authors study the non-locally bounded semimartingale setting without random endowment and unify the
framework for utilities on $\R$ and $\R_+$.

Stability with respect to perturbations in the market price of risk for fixed volatility is first studied in \cite{LZ07SPA} for utility on $\R_+$ and later in \cite{BK13SPA} for exponential utility.  Both works consider risky assets with continuous price processes and no random endowment.  For a locally bounded asset and an investor with random endowment, \cite{KZ11MF} study a market stability problem in which the financial market and random endowment stay fixed while the subjective probability measure and utility function vary.  A BSDE stability result is used in \cite{F13S} to study a specific stability problem for an exponential investor related to the indifference price formulas derived in \cite{FS08AAP}.  Using this BSDE stability result, \cite{F13S}'s market stability result extends to a case with a fixed market price of risk and a varying underlying correlation factor between the traded and nontraded securities.  In contrast to these previous works, we seek to prove a stability result for a general utility function on $\R$ allowing for varying both volatility and market price of risk in the presence of random endowment.

Stability 
is related to the concept of robustness with respect to a collection of probability measures.   
Robustness in option pricing dates back to the uncertain volatility models (UVM) of \cite{ALP95AMF} and \cite{Lyons95AMF}, who consider a range of possible volatilities and determine the best- and worst-case option prices.  In contrast to UVM, which seek to price claims in a \textit{complete} yet uncertain market, we seek to determine stability properties using indifference prices in an \textit{incomplete} market.  With utility maximization, both the volatility and the drift impact investors' optimal trading decisions.  In \cite{DK13SIAM}, \cite{MPZ15MF}, and \cite{TTU13FS}, the authors consider robust utility maximization problems, in which both the volatility and drift vary within a class of subjective probability measures.  Robust optimization seeks the best trading strategy in the worst possible model, whereas our investor firmly believes in the specified subjective model, and we seek to determine which of these models are stable.

The structure of the paper is as follows.  Section \ref{section:counterex} presents a counterexample for a power investor with unspanned stochastic volatility.  Section \ref{section:model} lays out the model assumptions and states the main result.  The proofs are presented in Section \ref{section:proofs}.  Finally, Section \ref{section:examples} provides a counterexample showing the necessity of a nondegeneracy assumption and provides sufficient conditions on the structure of the dual problem for this assumption to hold.

\section{Stability Counterexample for Power Utility}\label{section:counterex}
When an investor's preferences are described by utility on the positive real line and random endowment is present, the admissibility condition provides an additional implicit constraint.  As we will prove, this constraint can create a discontinuity in the value function and indifference prices for markets with varying martingale drivers.  The following are simple  incomplete Brownian models with a contingent claim that can only be replicated in the limiting market.

\subsection{Market Model}\label{subsection:counterex_model}
We let $B$ and $W$ be independent Brownian motions on a filtered probability space $\left(\Omega,\sF, \bF, \bP\right)$ where $\bF=(\sF_t)_{0\leq t\leq T}$ is the natural filtration of $(B, W)$ completed with $\bP$-null sets and $\sF=\sF_T$.  
We consider stock market models, $S^\rho$, with stochastic volatility indexed by correlation parameter $\rho\in(-1,1)$ where 

\begin{equation}\label{def:stocks}
\begin{split}
  dS^{\rho}_t &= \mu V_t dt 
    + \sqrt{V_t}\left(\sqrt{1-\rho^2}dB_t + \rho dW_t\right),
    \ \ \  S^{\rho}_0 :=0,\\
  dV_t &= \kappa\left(\theta-V_t\right)dt +  
    \sigma\sqrt{V_t}dB_t ,
    \ \ \  V_0 :=1.
\end{split}
\end{equation}  
The constants $\kappa, \theta,\sigma>0$ satisfy Feller's condition, $2\kappa\theta\geq \sigma^2$, which guarantees that there exists a unique strong solution for $V$ that is strictly positive for all $\rho\in(-1,1)$.  The risky asset $S^\rho$ is traded, whereas the stochastic volatility $V$ is not traded.  The dynamics of $S^\rho$ are written in an arithmetic fashion, which can be viewed as the returns of a positive asset.  For our purposes, the outcome of trading is unchanged whether we consider arithmetic or geometric specifications of the dynamics.  For a fixed $\rho$, \cite{K05QF} studies the utility maximization problem in the context of this model. Each $\rho$ market also has a bank account with zero interest rate.

A contingent claim $f$ is defined by $f:=\phi(B_T)$, where $\phi:\R\rightarrow\R$ is a bounded, continuous function.  The claim $f$ is replicable in the $\rho=0$ market; however, it is not replicable in any other market.  We  define $\phimin:=\inf\phi$, which corresponds to the subreplication price of $f$ in the $\rho\neq 0$ markets (see Proposition \ref{prop:subreplication} below).  We allow for the possibility that $\phi$ is a constant function, in which case the endowment $f$ can be viewed as a deterministic initial endowment.

\begin{remark}\label{rmk:measurability}
Our model assumes that all markets share the same probability space and filtration.  In particular, we assume that both Brownian motions, $B$ and $W$ are observable in each $\rho$ market.  However, suppose an investor in the $\rho$ market can only observe the path of the risky asset, $S^\rho$.  Then such an investor can also observe both $B$ and $W$.\footnote{Many thanks to an anonymous reviewer for making this keen observation.}  Since the quadratic variation of $S^\rho$ is observable and there exists a unique (positive) strong solution to the SDE for $V$, we can observe $V$ and $B$ from $\left<S^\rho\right>$.  Also from the observation of $S^\rho$ and $V$, we can determine $(\sqrt{1-\rho^2}B+\rho W)$, which allows the investor to observe both $B$ and $W$ separately (for $\rho\neq 0$).
\end{remark}

\subsection{Optimal Investment Problem}\label{subsection:counterex_primal}
An investor is modeled by power utility  $U(x)=x^p/p$ for $x\geq 0$ with $p<1$ ($p=0$ corresponds to $\log$).  As a convention, $U(x)=-\infty$ for $x < 0$.  The investor begins with initial capital $x>-\phimin$.  A progressively measurable process $H$ is \textit{integrable} if $\int_0^T V_t H^2_tdt<\infty$, a.s.  An integrable $H$ is called \textit{$\rho$-admissible} if there exists a finite constant $K=K(H)$ such that $(H\cdot S^\rho)_t\geq -K$ for all $t\in[0,T]$.  We define the primal optimization set by
  $$
    \sC(\rho):= \left\{(H\cdot S^\rho)_T: 
    \text{$H$ is $\rho$-admissible}\right\}.
  $$  
  For $\rho\in(-1,1)$, the primal value function is 
  defined by
\begin{equation}\label{def:primal}
  u(x,\rho):= \sup_{X\in\sC(\rho)} \E\left[U\left(x+X+f\right)\right],
  \ \ \text{ $x>-\phimin$}.
\end{equation}
\begin{remark}
  For $\rho=0$, $u(\cdot,0)$ is well-defined for a larger $x$-domain
  than $(-\phimin,\infty)$.  Yet the $x$-domain is tight
  for every $\rho\neq 0$.
  This discontinuity in the domains at $\rho=0$ hints at the issue of 
  (dis)continuity with respect to $\rho$ in the primal problem.  See 
  \cite{CSW01FS} for more details on the primal domain definition.
\end{remark}

For each $\rho\in(-1,1)$, we define the dual domain by
$$
  \sD(\rho):= \left\{\text{measures }\bQ\sim\bP: 
   \E\left[\frac{d\bQ}{d\bP}\right]=1 \text{ and } 
  \E^\bQ\left[X\right]\leq 0
  \  \forall X\in\sC(\rho)\right\}.
$$
  Lemma 5.2 in \cite{CL14RAPS} shows that $\sD(\rho)\neq\emptyset$.
Similar to \cite{EKQ95JCO}, we have the following result, which will be proven in Section \ref{section:proofs}.
\begin{proposition}\label{prop:subreplication}
Let $\rho\neq 0$ be given. The subreplication price of $f$ is $\phimin$; that is,
 $$
   \inf_{\bQ\in\sD(\rho)}\E^\bQ\left[\phi(B_T)\right] = \phimin.
 $$
 Moreover, for all $x\in\R$ and $(H\cdot S^\rho)_T\in\sC(\rho)$ such that 
 $x+(H\cdot S^\rho)_T+f\geq 0$, we have
 \begin{equation}\label{eqn:primal_bound}
   x+(H\cdot S^\rho)_T \geq -\phimin\,.
 \end{equation}
\end{proposition}

We consider a different optimization problem for $\rho=0$ with an additional admissibility constraint motivated by \eqref{eqn:primal_bound}.  For any $x>-\phimin$, we define the admissibly-constrained primal optimization sets in the $\rho=0$ market by 
$$
  \sC_c(x) := \left\{X\in\sC(0): 
  x+X \geq -\phimin \right\}.
$$
The corresponding admissibly-constrained primal value function is defined by
\begin{equation}\label{def:primal_constrained}
  u_c(x):= \sup_{X\in\sC_c(x)}\E\left[U\left(x+X+f\right)\right],
  \ \ \text{ $x>-\phimin$}.
\end{equation}

The following is the main result of the section.  {\label{sentence:no_f}We note that when $\phi(z)=0$ for all $z\in\R$, we have that $\sC(\rho=0)$ for $u(x,0)$ corresponds to $\sC_c(x)$, and $u(x,0) = u_c(x)$ for $x>-\phimin$.  In this case, the next theorem provides a stability result in the spirit of \cite{LZ07SPA}. }
\begin{theorem}\label{thm:primal_convergence}
  Assume the market dynamics \eqref{def:stocks} and utility function
  $U(x)=x^p/p$, for $x\geq 0$, with $p<1$ 
  ($p=0$ corresponds to $\log$).
  Assume the random endowment 
  function $\phi$ is continuous and bounded, and the initial 
  endowment is $x>-\phimin$.
  Let $u$ and $u_c$ be as in \eqref{def:primal} and 
  \eqref{def:primal_constrained}, respectively.  Then,
  $$
    \lim_{\rho\rightarrow 0} u(x,\rho) = u_c(x)\,.
  $$
\end{theorem}
The proofs of Theorem \ref{thm:primal_convergence} and its Corollary \ref{cor:indifference_prices} (below) will follow in Section \ref{section:proofs}.  The corollary says that when $\phi$ is not constant, indifference prices for $f$ do not converge to the unique arbitrage-free price in the $\rho=0$ market as $\rho\rightarrow 0$.  For any $\rho\in(-1,1)$, we define the value function without random endowment by 
\begin{equation}\label{def:primal_no_endowment}
  w(x,\rho):= \sup_{X\in\sC(\rho)} \E\left[U\left(x+X\right)\right],
  \ \ \text{ $x>0$}.
\end{equation}

\begin{definition}
  Given $x>-\phimin$ and $\rho\in(-1,1)$, $p=p(x,\rho)\in\R$ is called the
  \textit{indifference price for $f$ at $x$ in the $\rho$ market} if 
  $\, w(x+p,\rho) =  u(x,\rho)$.
\end{definition}
Of course, for $\rho=0$, the indifference price corresponds to the unique arbitrage-free price for the bounded replicable claim, $f$.  Also notice that since indifference prices are arbitrage-free prices, then $p(x,\rho)\geq\phimin$ for every $x>-\phimin$.  
\begin{corollary}\label{cor:indifference_prices}
  Under the assumptions of Theorem 
  \ref{thm:primal_convergence} and for $\phi$ not constant:  
  For $x>-\phimin$, the indifference
  prices for $f$ do not converge to the arbitrage-free price in the $\rho=0$
  market.  Indeed, $\limsup_{\rho\rightarrow 0} 
  p(x,\rho) < p(x,0)$.
\end{corollary}

\begin{remark}\label{rmk:counterex_simplicity}For the sake of clarity, emphasis is placed on the simplicity of the power investor's problem.  Some aspects can be generalized at the expense of more lengthy proofs and set-ups, e.g., a more general utility function or more general asset dynamics.  For the special case when $f=0$, Theorem \ref{thm:primal_convergence} can be extended to the more general market models of Section \ref{section:model} in order to generalize the value function convergence of Theorem 2.12 in \cite{LZ07SPA} in the varying volatility setting.
The difficulty in generalizing beyond $f=0$ stems from the need for a dual conjugacy result for $u_c$, which is not available in the literature due to the Inada condition not being satisfied at $x=0$ for the ($\omega$-dependent) function $x\mapsto U(x+f-\phimin)$.
\end{remark}

\section{Utility Functions on $\R$} \label{section:model}
Modeling investor preferences on the entire real line removes the fixed admissibility lower bound, which prevents the degeneracy of Theorem \ref{thm:primal_convergence} from occurring.  The remainder of this work is devoted to studying conditions that guarantee stability for real line utility functions.

Let $(\Omega,\sF,\bF=(\sF_t)_{0\leq t\leq T},\bP)$ be a filtered probability space with the filtration generated by $d$-dimensional Brownian motion $B=(B^1,\ldots, B^d)$.  We assume that $\bF$ is completed with all $\bP$ null sets and $\sF = \sF_T$, for a fixed time horizon $T\in(0,\infty)$.

We consider a sequence of financial market models with stocks $S^n$ valued in $\R$, for $1\leq n\leq \infty$,
  \begin{equation}\label{eqn:S_dynamics}
    dS^n_t = 
      \lambda^n_t\left|\sigma^n_t\right|^2 dt + \sigma^n_t dB_t\,,\ \ \ S^n_0 = 0.
  \end{equation}
Letting $\sL^p:=\{\text{progressively measurable }\theta: \int_0^T |\theta_t|^p dt<\infty,\,\text{a.s}\}$, $p=1,2$,  
we require that $\sigma^n=(\sigma^{n,1},\ldots,\sigma^{n,d})$ satisfies $\sigma^{n,i}\in\sL^2$ for 
$1\leq n\leq\infty$, $1\leq i\leq d$, and $\lambda^n\left|\sigma^n\right|^2\in\sL^1$ for $1\leq n\leq\infty$.  For $1\leq n\leq\infty$, we define the local martingales $M^n$ by
\begin{equation*}\label{eqn:Mn_dynamics}
  M^n := (\sigma^{n,1}\cdot B^1) + \ldots + (\sigma^{n,d}\cdot B^d),
\end{equation*}
so that the dynamics of $S^n$ are of the form
$$
  dS^n_t = \lambda^n_t d\left<M^n\right>_t + dM^n_t\,,
  \ \ \  S^n_0=0.
$$
Additionally, we assume that $\lambda^n\sigma^{n,i}\in\sL^2$ for $1\leq n\leq\infty$, $1\leq i\leq d$, so that $(\lambda^n\cdot M^n)$ is well-defined.  We let 
$Z^n_t := \sE\left(-\lambda^n\cdot M^n\right)_t$,  $t\in[0,T]$, denote each market's minimal martingale density process, where $\sE(\cdot)$ refers to the stochastic exponential.  Each market is assumed to have a bank account with a zero interest rate.
  
A sequence $\{X_n\}_{n\geq 1}$ of semimartingales is said to converge to $X$ in the \textit{semimartingale topology} provided that

$$
  \sup_{|\theta|\leq 1} 
  \E\left[\left|\left(\theta\cdot (X^n-X)\right)_T\right|\wedge 1\right]
  \longrightarrow 0 \ \text{ as } n\rightarrow\infty.
$$
Here, the supremum is taken over progressively measurable $\theta$, which are bounded uniformly by $1$ in $t$ and $\omega$.  We note that in the Brownian filtration, all progressively measurable processes are predictable.  
The following assumptions capture the necessary market regularity and the convergence of a sequence of markets.
\begin{assumption}\label{ass:semimg} 
  The collections $\{M^n\}_{1\leq n\leq\infty}$ and 
  $\{(\lambda^n\cdot M^n)\}_{1\leq n\leq\infty}$ satisfy the convergence relations:
  $$
    M^n\longrightarrow M^\infty\ \text{ and }\ 
    (\lambda^n\cdot M^n)\longrightarrow (\lambda^\infty\cdot M^\infty)
    \ \text{ in the semimartingale topology 
    as $n\rightarrow\infty$.}
  $$
\end{assumption}
The assumption that $(\lambda^n\cdot M^n)\longrightarrow
(\lambda^\infty\cdot M^\infty)$ is similar to the appropriate topology assumption of \cite{LZ07SPA}, whereas the convergence assumption on $M^n$ is new since the previous market stability work required the martingale components to remain constant.
\begin{assumption}\label{ass:Z_mg}
  Each minimal martingale density process, $Z^n$, for $1\leq n\leq\infty$,
  is a $\bP$-martingale.
\end{assumption}

Under the minimal martingale measure $\bQ^n$, where $\frac{d\bQ^n}{d\bP}=Z^n_T$, $S^n$ is a local martingale and any $\bP$-local martingale $N$ such that $\left<N, M^n\right>_t=0$ for $t\in[0,T]$ remains a local martingale under $\bQ^n$.  We refer to \cite{FS10} for a survey on minimal martingale measures and their use in mathematical finance.

  Under Assumption \ref{ass:semimg},  \label{rmk:about_semimg_convergence}
  $(\lambda^n\cdot M^n)_T\longrightarrow(\lambda^\infty\cdot M^\infty)_T$ 
  and $\left<\lambda^n\cdot M^n\right>_T \longrightarrow
  \left<\lambda^\infty\cdot M^\infty\right>_T$
  in probability as $n\rightarrow\infty$.  Hence,
  $Z^n_T \longrightarrow Z^\infty_T$ in probability as $n\rightarrow\infty$.  
  Under Assumption \ref{ass:Z_mg}, each $Z^n$ is a  
  martingale, and so Scheffe's Lemma implies the seemingly stronger fact that 
  $Z^n_T\longrightarrow Z^\infty_T$ in $L^1(\bP)$ as $n\rightarrow\infty$.

A further non-degeneracy assumption is needed on the limiting market  in order to ensure that randomness does not disappear in a degenerate way.  A counterexample showing that this condition is in some sense necessary is provided in Section \ref{section:examples}.
\begin{assumption}\label{ass:nondegeneracy}
  The dynamics of $\left<M^\infty\right>$ are nondegenerate in that
  $\sum_{i=1}^d\left(\sigma^{\infty,i}_t\right)^2\neq 0$ 
  for all $t\in[0,T]$, $\bP$-a.s.
\end{assumption}
\begin{remark}
  Assumptions \ref{ass:semimg}, \ref{ass:Z_mg},
  and \ref{ass:nondegeneracy} are satisfied
  by the markets $\{S^{\rho_n}\}_{1\leq n\leq\infty}$ of Section
  \ref{section:counterex} for any $\rho_n\longrightarrow\rho\in[-1,1]$
  as $n\rightarrow\infty$.
\end{remark}

Finally, a contingent claim $f\in L^\infty(\bP)$ is given and is independent of $n\in\bN$.  We make no assumption on the replicability of $f$ at this time.

\subsection{Optimal Investment Problem}\label{subsection:primal}

An investor is modeled by preferences $U:\R\rightarrow\R$, which is finite on the entire real line.  $U$ is assumed to be continuously differentiable, strictly increasing, strictly concave and satisfies the Inada conditions at $-\infty$ and $+\infty$:
\begin{equation}\label{eqn:Inada}
  U'(-\infty) := \lim_{x\rightarrow -\infty} U'(x) = \infty \ \text{ and }\ 
  U'(+\infty) := \lim_{x\rightarrow  \infty} U'(x) = 0.
\end{equation}
Additionally, we assume that $U$ satisfies the reasonable asymptotic elasticity conditions of \cite{KS99AAP} and \cite{S01AAP}:
  \begin{equation}\label{eqn:ae}
  AE_{-\infty}(U):= \liminf_{x\rightarrow-\infty} \frac{xU'(x)}{U(x)} > 1
  \ \text{ and }\ 
  AE_{+\infty}(U):= \limsup_{x\rightarrow\infty} \frac{xU'(x)}{U(x)} < 1\,.
  \end{equation}

The utility function's Fenchel conjugate is defined by $V(y):=\sup_{x\in\R}\left\{U(x)-xy\right\}$ for $y>0$.  $V$ is strictly convex and continuously differentiable.  Without loss of generality, we assume that $U(0)> 0$.  When $U(0)>0$, we have $V(y)> 0$ for all $y>0$.

Similar to \cite{LZ07SPA}, \cite{KZ11MF}, and \cite{BK13SPA}, we make the following assumption:
\begin{assumption}\label{ass:V_ui}
  The collection of random variables $\left\{V(Z^n_T)\right\}_{1\leq n\leq\infty}$,
  where $Z^n_T$ is the minimal martingale density for the $S^n$ market,
  is uniformly integrable.
\end{assumption}

In \cite{LZ07SPA}, the authors show that Assumption \ref{ass:V_ui} is both necessary and sufficient in the case of complete markets.  They study the stability problem with a utility function defined on the positive real line, no random endowment, fixed volatility, and varying market price of risk; see \cite{LZ07SPA} Proposition 2.13.  In an incomplete setting, they provide a counterexample to the value function stability showing that in some sense Assumption \ref{ass:V_ui} is necessary.

For $1\leq n\leq\infty$, a process $H$ is \textit{$S^n$-integrable} if $H\sigma^{n,i}\in\sL^2$ for $1\leq i\leq d$.  Cauchy-Schwartz's inequality produces $H\lambda^n(\sigma^{n,i})^2\in\sL^1$ for $1\leq i\leq d$.
The $S^n$ market's admissible strategies are defined by
$$
  \adm^n := \left\{ H: H \text{ is $S^n$-integrable},\ \exists K=K(H),\,
  (H\cdot S^n)_t\geq -K,\ \forall t\right\}.
$$
The primal value function is defined by
\begin{equation}\label{def:primal_value_fn}
  u_n(x) := \sup_{H\in\adm^n}\E\left[U\left(x+(H\cdot S^n)_T+f\right)\right], \ \ x\in\R.
\end{equation}

Let $\sM^n$ denote the set of probability measures $\bQ$ such that $\bQ\sim\bP$ and $S^n$ is a local martingale under $\bQ$.  We are primarily interested in such measures that have finite $V$-entropy:  $\E\left[V\left(\frac{d\bQ}{d\bP}\right)\right]<\infty$.  Let $\sM^n_V$ denote those measures $\bQ\in\sM^n$ having finite $V$-entropy.  
For $1\leq n\leq\infty$, the dual value function is defined for the $S^n$ market by
\begin{equation}\label{def:dual_value_fn}
  v_n(y) := \inf_{\bQ\in\sM^n_V} \E\left[V\left(y\frac{d\bQ}{d\bP}\right)
          + y\frac{d\bQ}{d\bP}f\right],\ \ y>0.
\end{equation}

\label{rmk:dual_is_same}
At first glance, our definition of the dual value function 
differs from that of \cite{OZ09MF}, who, for $1\leq n\leq\infty$, consider the infimum over $\bQ\ll\bP$ such that $S^n$ is a $\bQ$-local martingale and $\E[V(\frac{d\bQ}{d\bP})]<\infty$.  Assumptions \ref{ass:Z_mg} and \ref{ass:V_ui} plus $Z^n_T>0$  imply that $\sM^n_V\neq\emptyset$.  In this case, Theorem 1.1(iii) of \cite{OZ09MF} shows that the optimal dual element lies in the set $\sM^n_V$, and thus the two dual value function definitions agree. 

{\label{rmk:primal_is_same}
The primal admissible class of strategies is too small to attain a solution to the optimal investment problem.  However, the behavior of the value function is our primary interest, rather than the behavior (or even attainability) of the optimizer.  
Using that $f\in L^\infty(\bP)$ and $\sM^n_V\neq\emptyset$, Theorem 1.2(i) of \cite{OZ09MF} implies that our definition of the primal value function agrees with the definition of $u_\sE$ of \cite{OZ09MF}.  Here, $\sE=x_n+f$ and $\sE$ refers to the notation of the aforementioned work.
}

By using \label{page:AS_ref}\cite{AS92} and \cite{AS93MMM}, for $1\leq n\leq\infty$, we can rewrite any $\bQ\in\sM_V^n$ as $\frac{d\bQ}{d\bP} = Z_T^n\sE(L)_T$, where $L$ is a local martingale null at $0$ such that $\left<L,M^n\right>_t=0$ for all $t\in[0,T]$.
We need to make a further assumption in order to ensure a ``nice'' structure of the limiting market's dual domain.  For $n=\infty$, let $\sB$ be defined by
\begin{equation}\label{eqn:B}
\begin{split}
  \sB := \left\{\text{local martingales }L:\right. &  
    L_0=0, \left<L, M^\infty\right>_t = 0,\,\forall t\in[0,T], \,
    \\ &\left. \exists \text{ constant }C=C(L),\, \sE(L)_t\leq C,
    \,\forall t\in[0,T]\right\}.
\end{split}
\end{equation}

\begin{assumption}\label{ass:bdd_inf}
  For $n=\infty$, the dual problem, (\ref{def:dual_value_fn}),
  can be expressed as
  $$
    v_\infty(y) = \inf_{L\in\sB} \E\left[V\left(yZ^\infty_T\sE(L)_T\right)
      + yZ^\infty_T\sE(L)_T f\right],\ \ y>0,
  $$
  where $Z^\infty_T$ is the minimal martingale density in the 
  $S^\infty$ market.
\end{assumption}

This assumption is non-trivial to verify in general due to the fact that $V$ is increasing strictly faster than linearly as  $y\longrightarrow +\infty$.  It is mathematical in nature and ensures that the dual optimizer does not vary ``too much''.  Section \ref{section:examples} provides two sufficient conditions. The first condition covers the original motivation for our stability problem, where the contingent claim is replicable in the (incomplete) limiting market but not replicable in any pre-limiting market.  In this case, the limiting market consists of a driving Brownian motion, a replicable claim, and additional independent Brownian noise.  The second condition makes no assumptions on the claim's replicability; however, it requires exponential preferences and imposes a mild $\bmo$ condition on the limiting market.  Indeed, a $\bmo$ condition on the limiting market's minimal martingale density ensures that the dual optimizer has controlled oscillations, which implies Assumption \ref{ass:bdd_inf}.  Similarly, \cite{6AP} make use of a form of $\bmo$ regularity of \textit{some} dual element in order to establish $\bmo$ regularity of the \textit{optimal} dual element.

The following is our main result.
\begin{theorem}\label{thm:main_result}
  Suppose that the sequence of markets satisfies Assumptions 
  \ref{ass:semimg}, \ref{ass:Z_mg},
  and \ref{ass:V_ui}.  Suppose that 
  the limiting market satisfies
  Assumptions \ref{ass:nondegeneracy} and \ref{ass:bdd_inf}.  
  Then, for $x_n\longrightarrow x$ as 
  $n\rightarrow\infty$,
  $$
    \lim_{n\rightarrow\infty} u_n(x_n) = u_\infty(x).
  $$
  Moreover, for $y_n\longrightarrow y>0$ as 
  $n\rightarrow \infty$,
  $$
    \lim_{n\rightarrow\infty} v_n(y_n) = v_\infty(y).
  $$
\end{theorem}

For $1\leq n\leq\infty$, the value function without random endowment is defined by
\begin{equation}
  w_n(x) := \sup_{H\in\adm^n}\E\left[U\left(x+(H\cdot S^n)_T\right)\right],
  \ \ x\in\R.
\end{equation}
\begin{definition}\label{def:indifference_price}
  Given $1\leq n\leq\infty$ and $x\in\R$, $p_n=p_n(x)$ is called the \textit{indifference price for
  $f$ at $x$ in the $S^n$ market} if $w_n(x+p_n) = u_n(x)$.
\end{definition}
\begin{corollary}\label{cor:indifference_price}
  Let the assumptions be as in Theorem \ref{thm:main_result}.
  Then for $x\in\R$, the indifference prices for $f$ converge;
  that is, $\lim_{n\rightarrow\infty} p_n(x) = p_\infty(x)$.
\end{corollary}

\begin{remark}
  The results in Theorem \ref{thm:main_result} and 
  Corollary \ref{cor:indifference_price} remain true 
  (with only minor notational changes to the proofs)
  in the case with varying random endowment.  Specifically, the random
  endowments $\{f_n\}_{1\leq n\leq\infty}$ corresponding to the 
  $\{S^n\}_{1\leq n\leq\infty}$ markets need to satisfy
  \begin{equation}\label{eqn:fn}
    \sup_n\|f_n\|_{L^\infty} <\infty\ \ \text{ and }\ \ 
    f_n\longrightarrow f_\infty \ \text{ in probability as $n\rightarrow\infty$}
  \end{equation}
  in order for the results to hold.  This additional flexibility allows us to
  consider the case of a varying quantity of contingent claims and also
  contingent claims that depend on the individual markets.  For example,
  if $g:\R\rightarrow\R$ is bounded and continuous, then $f_n:=g(S^n_T)$
  will satisfy \eqref{eqn:fn}.
\end{remark}

\begin{remark}\label{rmk:optimal_wealths}
  The study of the optimal terminal wealths and the optimal dual elements is
  typical in utility maximization in addition to properties of
  the value functions; 
  however, it is absent in the present work.  
  When varying both the volatility and drift of the risky assets, a major
  hurdle to stability is handling the change in the primal and dual feasible
  elements from market to market.
  Here, we use the varying volatility as a tool for pricing
  financial securities via a ``nearby'' models with good properties,
  rather than using it for investment advice.
  Because the study of optimal strategies is rather involved, it is beyond 
  the scope of the present work and would be an interesting question to 
  address in future research.
\end{remark}

\begin{remark}\label{rmk:kz_comparison}
  A special case of stability with varying volatility is considered in
  \cite{KZ11MF} (see their Remark 2.5).  The authors consider a fixed 
  risky asset with varying \textit{equivalent} subjective probability
  measures.  However, this approach relies on the invertibility of the 
  volatility process in every market, which in particular implies completeness 
  for all markets.  
  In our model, such measures would correspond to the risky asset laws, 
  $\bP^n:=\bP\circ (S^n)^{-1}$.  Due to the changes in the volatility 
  structure with $n$, the laws $\bP^n$ are \textit{nonequivalent} in the 
  present work.  Moreover, our results do not rely on completeness.
\end{remark}

\section{Proofs}\label{section:proofs}
We begin by proving the results from Section \ref{section:counterex} for the power investor.

\subsection{Dual Problems and Power Investor Proofs}\label{subsection:counterex_dual}
We begin by proving Proposition \ref{prop:subreplication}.  Example 1 of \cite{EKQ95JCO} uses the duality between $L^\infty(\bP)$ and $L^1(\bP)$ in order to establish a similar result when the contingent claim is independent of the traded assets.  Without independence, we cannot apply the duality result directly, and instead we explicitly construct a sequence of martingale measures realizing the subreplication price.
\begin{proof}[Proof of Proposition \ref{prop:subreplication}]
Let $\rho\neq 0$ be given.  We first seek to show that for all $0<t'<T$, 
$$
  \essinf_{\bQ\in\sD(\rho)} \E^\bQ[\phi(B_T)|\sF_{t'}] = \phimin,
$$ 
which implies that the subreplication price is $\phimin$.  Subsequently, we will show \eqref{eqn:primal_bound}.

We fix $t'<T$ and let $T'\in (t',T)$ and $x\in\R$ be given.  Then $B^\rho:= \sqrt{1-\rho^2} B + \rho W$ and $W^\rho:= \sqrt{1-\rho^2} W - \rho B$ are orthogonal $\bP$-Brownian motions.  Equivalently, we have $B = \sqrt{1-\rho^2}B^\rho -\rho W^\rho$ and $W=\rho B^\rho+\sqrt{1-\rho^2}W^\rho$.  
Consider the local martingale $Z$ defined for $t\in[0,T]$ by
$$
  Z_t := \sE\left(-\mu\sqrt{V}\cdot B^\rho\right)_t
  \sE\left(\frac{1}{\rho}\left(-\mu\sqrt{1-\rho^2}\sqrt{V}-\frac{x}{T}
  +\frac{\eta\bI_{[T',T]}}{T-T'}\right)\cdot W^\rho\right)_t,
$$
where  $\eta:= B_{T'} - xT'/T\in\sF_{T'}$.  
In fact, $Z$ is a martingale, which we verify by applying Novikov's condition locally.  The following procedure is standard; see, e.g., Section 6.2 Example 3(a) in \cite{LS01}.  By Corollary 5.14 of \cite{KS91}, it suffices to find $\Delta>0$ and $t_n:= n\Delta$ such that for each $n\geq 1$,
\begin{equation}\label{eqn:novikov}
  \E\left[\exp\left(\frac{1}{2}\int_{t_n}^{t_{n+1}}d\left<M\right>_u\right)\right]
  <\infty,
\end{equation}
where $M_t:=-\mu\left(\sqrt{V}\cdot B^\rho\right)_t + 
  \frac{1}{\rho}\left(\left(-\mu\sqrt{1-\rho^2}\sqrt{V}-\frac{x}{T}
  +\frac{\eta\bI_{[T',T]}}{T-T'}\right)\cdot W^\rho\right)_t$ for $t\in[0,T]$.  By applying Cauchy-Schwartz to \eqref{eqn:novikov}, it suffices to choose $\Delta>0$ such that for each $n\geq 1$, we have
  \begin{equation}\label{eqn:novikov2}
    \E\left[\exp\left(\int_{t_n}^{t_{n+1}}\frac{\mu^2}{\rho^2}V_u du\right)\right]
  <\infty \ \ \text{ and } \ \ 
  \E\left[\exp\left(\Delta\frac{\eta^2}{\rho^2(T-T')}\right)\right]
  <\infty.
  \end{equation}
  Jensen's Inequality and Tonelli's Theorem imply that $\E\left[\exp\left(\int_{t_n}^{t_{n+1}}\frac{\mu^2}{\rho^2}V_u du\right)\right] \leq \\ \E\left[\int_{t_n}^{t_{n+1}}\exp\left(\frac{\Delta\mu^2}{\rho^2}V_u \right)\frac{du}{\Delta}\right] \leq \E\left[\exp\left(\frac{\mu^2}{\rho^2}\Delta V_T\right)\right]$. Thus, taking $\Delta := \rho^2\min\left(\frac{\kappa}{\mu^2\sigma^2(1-e^{-\kappa T})},\frac{T-T'}{4T'}\right)$ yields \eqref{eqn:novikov2}.

We define $\overline{\bQ}\in\sD(\rho)$ by $\frac{d\overline\bQ}{d\bP}:= Z_T$ and the processes $\overline{B}^\rho$ and $\overline{W}^\rho$ by
$$
  \overline{B}^\rho_t := B^\rho_t +\mu\int_0^t \sqrt{V_u}\,du
$$
and
$$
  \overline{W}^\rho_t := W^\rho_t + \frac{\mu\sqrt{1-\rho^2}}{\rho}\int_0^t \sqrt{V_u}\,du + \frac{x t}{\rho T} - \frac{\eta \int_0^t\bI_{[T',T]}}{\rho(T-T')}.
$$
By Girsanov's Theorem, $\overline{B}^\rho$ and $\overline{W}^\rho$ are orthogonal $\overline\bQ$-Brownian motions.  
Moreover, $\eta = \sqrt{1-\rho^2}\overline{B}^\rho_{T'} - \rho\overline{W}^\rho_{T'}$, which implies that
$$
  B_T = \left(\sqrt{1-\rho^2}\overline{B}^\rho_T-\rho\overline{W}^\rho_T\right) + x - \left(\sqrt{1-\rho^2}\overline{B}^\rho_{T'}-\rho\overline{W}^\rho_{T'}\right).
$$
Then, $\bP$-a.s.,
\begin{align*}
  \E^{\overline{\bQ}}\left[\phi(B_T)\left.\right|\sF_{t'}\right]
  &= \E^{\overline{\bQ}}\left[\phi\left(\left(\sqrt{1-\rho^2}\overline{B}^\rho_T-\rho\overline{W}^\rho_T\right) - \left(\sqrt{1-\rho^2}\overline{B}^\rho_{T'}-\rho\overline{W}^\rho_{T'}\right) +x\right)\left.\right|\sF_{t'}\right]\\
  &= \E^{\overline{\bQ}}\left[\phi\left(\left(\sqrt{1-\rho^2}\overline{B}^\rho_T-\rho\overline{W}^\rho_T\right) - \left(\sqrt{1-\rho^2}\overline{B}^\rho_{T'}-\rho\overline{W}^\rho_{T'}\right) +x\right)\right]\\ 
  &= \E^\bP\left[\phi\left(B_T-B_{T'}+x\right)\right].
\end{align*}

The choice of $T'\in(t',T)$ and $x\in\R$ is arbitrary, and therefore,
\begin{equation}\label{eqn:subreplication_confirmed}
  \essinf_{\bQ\in\sD(\rho)}\,\E^\bQ\left[\phi(B_T)\left.\right|\sF_{t'}\right]
  =\phimin.
\end{equation}

Finally, we suppose that $x\in\R$ and $(H\cdot S^\rho)_T\in\sC(\rho)$ such that $x+(H\cdot S^\rho)_T+\phi(B_T)\geq 0$.  Then for all $\bQ\in\sD(\rho)$, we have that $(H\cdot S^\rho)$ is a lower-bounded $\bQ$-local martingale, and hence a $\bQ$-supermartingale.  For all $t'<T$, we have $0\leq x+(H\cdot S^\rho)_{t'} + \E^\bQ\left[\phi(B_T)|\sF_{t'}\right]$.  By \eqref{eqn:subreplication_confirmed} above, we have $0\leq x+ (H\cdot S^\rho)_{t'}+\phimin$.  Continuity with respect to time and taking $t'\rightarrow T$ yields \eqref{eqn:primal_bound}.
\end{proof}

As is typical in convex optimization, we introduce the dual problem as tool for proving Theorem \ref{thm:primal_convergence} and Corollary \ref{cor:indifference_prices}.  For $y>0$, define $V(y):=\sup_{x>0}\left\{U(x)-xy\right\}$.  For $U(x)=x^p/p$, we have $V(y) = \frac{1-p}{p}y^{p/(p-1)}$.  
For $y>0$ and $z\geq \phimin$, we define
$$
  V_c(y,z):=\sup_{x>-\phimin}\left\{U(x+z)-xy\right\}=
    \begin{cases}
      V(y)+yz, & \text{for } y< U'\left(z-\phimin\right), \\
      U\left(z-\phimin\right)+y\phimin, & \text{otherwise}.
    \end{cases}
$$
We can then define a constrained form of the dual value function for $\rho\in(-1,1)$ by,
\begin{equation}\label{def:constrained_v}
  v_c(y,\rho):= \inf_{\bQ\in\sD(\rho)}
    \E\left[V_c\left(y\frac{d\bQ}{d\bP},f\right)\right],
    \ \ \text{ $y>0$}.
\end{equation}
For $Z^\rho_t := \sE\left(-\mu\sqrt{V}\cdot B\right)_t$, $t\in[0,T]$, the random variable $Z^\rho_T$ is the minimal martingale density corresponding to the $S^\rho$ market.  The martingale property of $Z^\rho$ is shown in Lemma 5.2 of \cite{CL14RAPS}.  In particular,   $v_c(y,\rho)<\infty$ for all $y>0$ and $\rho\in(-1,1)$.

The constrained dual problem arises naturally from the endogenous primal admissibility constraint \eqref{eqn:primal_bound}.  For $\rho\neq 0$, we could define a constrained primal problem, $u^\rho_c=u^\rho_c(x)$ for $x>-\phimin$, and a corresponding constrained optimization set, $\sC^\rho_c(x)$, analogously to $u_c$ and $\sC_c(x)$ in the $\rho=0$ case.  In that case, we would have $u^\rho_c(x) = u(x,\rho)$ for all $x>-\phimin$ by \eqref{eqn:primal_bound}, and \eqref{def:constrained_v} would be the natural candidate for its dual conjugate.  Indeed, for $\rho\neq 0$, \cite{LSZ14} prove that the constrained form of the
  dual value function, \eqref{def:constrained_v}, 
  is in fact equal to the dual value function
  as it is defined in \cite{CSW01FS}, Equation (3.1).  
  (See \cite{LSZ14} Theorem 4.2.)

\begin{remark}\label{rmk:facelifting}
    In \cite{LSZ14},
  the authors consider the 
  problem of \textit{facelifting}, in which the primal and dual value
  functions in the presence of unspanned random endowment are not continuous with respect to time to maturity as the maturity decreases to $0$. 
  At first glance, our stability problem differs from that of varying
  maturity.  However, both problems have the property that the random 
  endowment is non-replicable in every pre-limiting market yet replicable 
  in the limit.  This property allows for the admissibility 
  constraint, \eqref{eqn:primal_bound}, to appear endogenously 
  in the pre-limiting 
  models, whereas \eqref{eqn:primal_bound} must be exogenously applied 
  in the limiting model.
\end{remark}

\begin{lemma}\label{lemma:v_usc}
  Let the assumptions of the model be as in Theorem 
  \ref{thm:primal_convergence}.
  For $y>0$, 
  $$
    \limsup_{\rho\rightarrow 0} v_c(y,\rho)\leq \E\left[V_c(yZ^0_T,f)\right],
  $$
  where $Z^0_T$ is the minimal martingale density for the $S^0$ market.
\end{lemma}

\begin{proof}
  One can show that $\{V(yZ^\rho_T)\}_\rho$ is uniformly integrable, 
  e.g., using the proof of Lemma 5.2 in \cite{CL14RAPS}.  
  For each $\rho\in(-1,1)$, $Z^\rho$ is a martingale, and hence 
  convergence in probability along with Scheffe's
  Lemma implies that $Z^\rho_T\longrightarrow Z^0_T$ in $L^1(\bP)$ as
  $\rho\rightarrow 0$.  Convergence in $L^1(\bP)$ plus $f\in L^\infty(\bP)$
  implies that $\{Z^\rho_T f\}_\rho$ is 
  uniformly integrable.  Since $V_c(yZ^\rho_T,f)\longrightarrow V_c(yZ^0_T,f)$
  in probability as $\rho\rightarrow 0$ and
  $$
    V_c(yZ^\rho_T,f) \leq V(yZ^\rho_T)+yZ^\rho_Tf
  $$
  for all $\rho\in(-1,1)$, Fatou's Lemma implies
  \begin{align*}
    \E\left[V_c\left(yZ^0_T,f\right)\right]
      &\geq \limsup_{\rho\rightarrow 0} 
        \E\left[V_c\left(yZ^\rho_T,f\right)\right] \\
      &\geq \limsup_{\rho\rightarrow 0} \ v_c(y,\rho). 
  \end{align*}
\end{proof}

\begin{lemma}\label{lemma:u_lsc}  
  Let the assumptions of the model be as in Theorem 
  \ref{thm:primal_convergence}.
  Let $u$ and $u_c$ be as defined in
  \eqref{def:primal} and \eqref{def:primal_constrained}, respectively.
  For any $x>-\phimin$, $u_c(x)\leq\liminf_{\rho\rightarrow 0} u(x,\rho)$.
\end{lemma}

Before proving Lemma \ref{lemma:u_lsc}, we need two technical lemmas, which will again be used in the proof of Theorem \ref{thm:main_result}.  While the notions of integrability are defined separately for Sections \ref{section:counterex} and \ref{section:model}, the notions agree 
and are not referred to separately in Lemmas \ref{lemma:H_approx1} and \ref{lemma:H_approx2} below.
\begin{lemma}\label{lemma:H_approx1}
  Let $X$ be a semimartingale and $H$ be $X$-integrable.  Suppose that there 
  exists $K>0$ such that $(H\cdot X)_t\geq -K$ for all $t\in[0,T]$.  
  Then for any $\delta>0$ there exists a sequence of progressively measurable 
  integrands $\{H^n\}_{n\geq 1}$ such that for each $n\geq 1$, $H^n$ is 
  uniformly bounded in $t$ and $\omega$, while for all $t\in[0,T]$,
  $$
    (H^n\cdot X)_t\geq -K-\delta,
  $$
  and $(H^n\cdot X)_T\longrightarrow (H\cdot X)_T$ in probability 
  as $n\rightarrow\infty$.
\end{lemma}
\begin{proof}
  For $n\geq 1$, we define the integrands $H^n:=H\bI_{\{|H|\geq n\}}$, where 
  $\bI_A$ denotes the indicator function of a set $A\subset\Omega\times[0,T]$. 
  We define the stopping times
  $$
    \sigma_n:=\inf\left\{t\leq T: (H^n\cdot X)_t\leq -K-\delta\right\}.
  $$
  Then $(H^n\bI_{[0,\sigma_n]}\cdot X)_t\geq -K-\delta$ for all 
  $t\in[0,T]$.  Moreover, $\sup_t|((H^n\bI-H)\cdot X)_t|\longrightarrow 0$ 
  in probability as $n\rightarrow\infty$ by Lemma 4.11 and Remark (ii) 
  following Definition 4.8 both in \cite{CS02}.  This convergence implies that
  $\bP(\sigma_n=T)\longrightarrow 1$ and hence 
  $(H^n\cdot X)_{\sigma_n}\longrightarrow (H\cdot X)_T$ in probability 
  as $n\rightarrow\infty$.  
  Considering the sequence $\{H^n\bI_{[0,\sigma^n]}\}_{n\geq 1}$ yields 
  the result.
\end{proof}
\begin{lemma}\label{lemma:H_approx2}
  Let $\{X^n\}_{n\geq 1}$ be a sequence of semimartingales such that
  $X^n\longrightarrow X$ in the semimartingale topology as 
  $n\rightarrow\infty$.  Suppose that $H$ is progressively measurable and
  uniformly bounded in $t$ and $\omega$ and there exists a $K>0$ such 
  that $(H\cdot X)_t\geq -K$ for all $t\in[0,T]$.  Then for any $\delta>0$, 
  there exists a sequence $\{H^n\}_{n\geq 1}$ such that for all $n\geq 1$, 
  $H^n$ is uniformly bounded in $t$ and $\omega$, 
  for all $t\in[0,T]$,
  $$
    (H^n\cdot X^n)_t \geq -K -\delta,
  $$
  and $(H^n\cdot X^n)_T \longrightarrow (H\cdot X)_T$ in probability 
  as $n\rightarrow\infty$.
\end{lemma}
\begin{proof}
  Since $H$ is uniformly bounded and progressively measurable, it is 
  $X$- and $X^n$-integrable for all $n\geq 1$.  Moreover, the definition 
  of semimartingale convergence implies that
  \begin{equation}\label{eqn:semimg_integral_convergence}
    (H\cdot X^n)\longrightarrow (H\cdot X) 
    \ \text{ in the semimartingale topology as $n\rightarrow\infty$}.
  \end{equation}
  For $n\geq 1$, we define the stopping times $\tau_n$ by
  $$
    \tau_n:=\inf\left\{t\leq T: (H\cdot X^n)_n< -K-\delta\right\}
  $$
  and let $H^n:= H\bI_{[0,\tau_n]}$.  By definition of $\tau_n$, we have 
  $(H^n\cdot X^n)_t\geq -K-\delta$ for all $t\in[0,T]$.  Using 
  \eqref{eqn:semimg_integral_convergence},
  \begin{align*}
    \bP(\tau_n<T)
      &= \bP\left(\exists\ t'< T: (H\cdot X^n)_{t'}< -K-\delta\right)\\
      &\leq \bP\left(\sup_{t\leq T}\left|\left(H\cdot(X^n-X)\right)_t\right|
        > \delta\right) + 
        \bP\left(\exists\ t'\leq T: (H\cdot X)_{t'}<-K\right) \\
      &= \bP\left(\sup_{t\leq T}\left|\left(H\cdot(X^n-X)\right)_t\right|
        > \delta\right) + 0 \\
      &\longrightarrow 0 \ \text{ as $n\rightarrow\infty$.}
  \end{align*}
  Thus, $(H^n\cdot X^n)_T = (H\cdot X^n)_{\tau_n}\longrightarrow (H\cdot X)_T$  
  in probability as $n\rightarrow\infty$.
\end{proof}

\begin{proof}[Proof of Lemma \ref{lemma:u_lsc}]
Let $\eps>0$ be given.  Since $u_c$ is concave, it is continuous on the interior of its domain, and hence we may choose $x'<x$ such that $u_c(x)<u_c(x')+\eps$.  We choose $(H\cdot S^0)_T\in\sC_c(x')$ such that $u_c(x')\leq\E\left[U\left(x'+(H\cdot S^0)_T +f\right)\right]+\eps$.

We define $\delta:=\frac{x-x'}{4}>0$.  Since $H$ is $(\rho=0)$-admissible and $x'+(H\cdot S^0)_T\geq -\phimin$, Lemma \ref{lemma:H_approx1} provides us with a sequence of integrands $\{H^n\}_{n\geq 1}$ such that for each $n\geq 1$, $H^n$ is bounded uniformly in $t$ and $\omega$ while $x'+\delta+(H^n\cdot S^0)_t\geq -\phimin$ for all $t\in[0,T]$.  In particular, for all $n\geq 1$, $(H^n\cdot S^0)_T\in\sC_c(x'+\delta)\subseteq\sC_c(x)$, and we have the uniform lower bound 
$$
  U\left(x+(H^n\cdot S^0)_T+f\right)
  \geq U\left(x-x'-\delta+f-\phimin\right)
  \geq U\left(\frac{3}{4}(x-x')\right) > -\infty.
$$
Fatou's Lemma implies that
\begin{align*}
  u_c(x) &\leq u_c(x')+\eps \\
  &\leq \E\left[U\left(x+(H\cdot S^0)_T +f\right)\right] + 2\eps \\
  &\leq \liminf_{n\rightarrow\infty} \E\left[U\left(x+(H^n\cdot S^0)_T +f\right)\right] + 2\eps,
\end{align*}
which allows us to choose a sufficiently large $n$ such that the integrand $\tilde{H}:= H^n$ is uniformly bounded in $t$ and $\omega$, $(\tilde{H}\cdot S^0)_T\in \sC_c(x'+\delta)$ and
\begin{equation}\label{eqn:uc_approx}
  u_c(x) 
  \leq \E\left[U\left(x+(\tilde{H}\cdot S^0)_T +f\right)\right] + 3\eps.
\end{equation}

Now that we have achieved sufficiently nice regularity of a nearly-optimal strategy, $\tilde{H}$, we proceed by varying the parameter $\rho$.  Let $\rho_k\longrightarrow 0$ be a sequence realizing the $\liminf$ in $\liminf_{\rho\rightarrow 0} u(x,\rho)$.  Since $S^{\rho_k}\longrightarrow S^0$ in the semimartingale topology as $k\rightarrow\infty$, Lemma \ref{lemma:H_approx2} allows us to choose $\{\tilde{H}^k\}_{k\geq 1}$ such that for each $k\geq 1$, $\tilde{H}^k$ is bounded uniformly in $t$ and $\omega$ while $x'+2\delta + (\tilde{H}^k\cdot S^{\rho_k})_t\geq -\phimin$.  Moreover, $(\tilde{H}^k\cdot S^{\rho_k})_T\longrightarrow (\tilde{H}\cdot S^0)_T$ in probability as $k\rightarrow\infty$.  For every $k\geq 1$, we have the uniform lower bound
$$
  U(x+(\tilde{H}^k\cdot S^{\rho_k})_T+f)\geq U(x-x'-2\delta+f-\phimin)\geq U(\frac{1}{2}(x-x'))>-\infty.
$$
  Therefore by Fatou's Lemma and \eqref{eqn:uc_approx} above,
\begin{align*}
  u_c(x) 
  &\leq \E\left[U\left(x+(\tilde{H}\cdot S^0)_T +f\right)\right] + 3\eps \\
  &\leq \liminf_{k\rightarrow\infty} \E\left[U\left(x'+(\tilde{H}^k\cdot 
    S^{\rho_k})_T +f\right)\right] + 3\eps \\
  &\leq \liminf_{k\rightarrow\infty} \E\left[U\left(x+(\tilde{H}^k\cdot 
    S^{\rho_k})_T +f\right)\right] + 3\eps \\
  &\leq \liminf_{k\rightarrow\infty} u(x,\rho_k) + 3\eps \\
  &= \liminf_{\rho\rightarrow 0} u(x,\rho) + 3\eps.
\end{align*}
Since $\eps>0$ is arbitrary, the desired result holds.
\end{proof}

\begin{proof}[Proof of Theorem \ref{thm:primal_convergence}]
  Fix $\rho\neq 0$.  For $x>-\phimin$, $X\in\sC(\rho)$ such that
  $x+X\geq -\phimin$, $y>0$, and $\bQ\in\sD(\rho)$, we have
  \begin{align*}
    \E\left[U(x+X+f)\right]
      &\leq \E\left[V_c\left(y\frac{d\bQ}{d\bP},f\right)
      + y\frac{d\bQ}{d\bP}\left(x+X\right)\right] \\
      &\leq \E\left[V_c\left(y\frac{d\bQ}{d\bP},f\right)\right]+xy.
  \end{align*}
  This strengthening of Fenchel's inequality relies on the bound
  $x+X\geq-\phimin$ in order to replace $V$ with $V_c(\cdot,f)$.
  Next, we take the supremum over all $X\in\sC(\rho)$ with $x+X\geq-\phimin$
  and the infimum over all $\bQ\in\sD(\rho)$, which yields
  that for any $x>-\phimin$ and $y>0$,
  $$
    u(x,\rho)\leq v_c(y,\rho)+xy.
  $$
  This inequality along with Lemmas \ref{lemma:v_usc} and
  \ref{lemma:u_lsc} shows that for any $x>-\phimin$ and $y>0$,
  \begin{equation}\label{eqn:inequality}
    u_c(x)\leq\liminf_{\rho\rightarrow 0} u(x,\rho)
    \leq \limsup_{\rho\rightarrow 0} v_c(y,\rho) + xy
    \leq \E[V_c(yZ^0_T,f)] + xy\,.
  \end{equation}
  Next, we show that $u_c(\cdot)$ and $v_c(\cdot,0)$ are conjugates.  
  We let $y>0$ be given and define the candidate optimal
  terminal wealth $\Xhat$ by
  $$
    \Xhat :=
    \begin{cases}
      -V'(yZ^0_T)-f\,, & \text{ if } yZ^0_T\leq U'(f-\phimin),\\
      -\phimin\,,    & \text{ otherwise}.
    \end{cases}
  $$
  For $\frac{d\bQ^0}{d\bP}:=Z^0_T=\sE(-\mu\sqrt{V}
  \cdot B)_T$, 
  we have that $\hat{X}\in L^2(\bQ^0)$.
  By martingale representation and the strict positivity of $\sqrt{V}$ by the Feller condition,
  we may write $\Xhat = \E^{\bQ^0}[\Xhat]+ (H\cdot S^0)_T$ for some 
  integrable $H$.  
  Since $\Xhat\geq -\phimin$ and $(H\cdot S^0)$ is a $\bQ^0$-martingale, 
  we know that $(H\cdot S^0)_t\geq -\phimin-\E^{\bQ^0}[\Xhat]>-\infty$ 
  for all $t\in[0,T]$.   Thus, $H$ is $S^0$-admissible.
  
  We define $\xhat:= \E^{\bQ^0}[\Xhat]>-\phimin$ 
  so that $\Xhat-\xhat\in\sC_c(\xhat)$.
  For any $y>0$,
  \begin{align*}
    \E\left[V_c\left(yZ^0_T,f\right)\right] 
      &= \E\left[ U\left(\Xhat+f\right)-yZ^0_T\Xhat\right] \\
      &= \E\left[U\left(\Xhat+f\right)\right] - y\xhat \\
      &\leq \sup_{x>-\phimin}\left\{\sup_{X\in\sC_c(x)}
        \E\left[U(x+X+f)\right]-yx\right\} \\
      &= \sup_{x>-\phimin}\left\{u_c(x)-xy\right\}\,.
  \end{align*}
  Since the other direction of the inequality holds by \eqref{eqn:inequality}, 
  we obtain that for any $y>0$,
  $\E[V_c(yZ^0_T,f)] = \sup_{x>-\phimin}\left\{u_c(x)-xy\right\}$.
  Since $u_c(\cdot)$ is concave and upper 
  semicontinuous on $(-\phimin,\infty)$, we have 
  $u_c(x) = \inf_{y>0}\left\{\E[V_c(yZ^0_T,f)]+xy\right\}$ for $x>-\phimin$.
  Strict convexity of
  $y\mapsto \E[V_c(yZ^0_T,f)]$ implies the differentiability of
  $u_c(\cdot)$.  (See, e.g., Proposition 6.2.1 on page 40 of
  \cite{HUL96}.)
  Now for any $x>-\phimin$, choosing 
  $y=\frac{d}{dx}u_c(x)$ yields equality in
  \eqref{eqn:inequality}.
\end{proof}

Finally, we show that indifference prices do not converge as 
$\rho\rightarrow 0$.
\begin{proof}[Proof of Corollary \ref{cor:indifference_prices}]
  Let $x>-\phimin$ be given.  For any $\rho\in(-1,1)$,
  $w(x,\rho)=w(x,0)$.
  Suppose that for $\rho_n\longrightarrow 0$,
  we have $p(x,\rho_n)\longrightarrow \bar{p}$ as $n\rightarrow\infty$.
  By being the limit of arbitrage-free prices in the 
  $\{\rho_n\}_n$ models, we must have 
  $\bar{p}\in [\inf\phi,\sup\phi]$.
  
  Using that $\phi$ is not constant, 
  for $x>-\phimin$, we first note that $u_c(x)<u(x,0)$.  This result
  can be obtained, for example, by Theorem 2.2 of \cite{KS99AAP}
  and $f$'s replicability in the $S^0$ market, which imply that
  $u(x,0)=\E\left[U(I(\frac{\partial}{\partial x}u(x,0)Z^0_T))\right]$ where 
  $\bP\left(I(\frac{\partial}{\partial x}u(x,0)Z^0_T)
  <f-\phimin\right)>0$.
  By Theorem \ref{thm:primal_convergence},
  $$
    \lim_n u(x,\rho_n) = u_c(x) < u(x,0) = w(x+p(x,0),0).
  $$
  Taking $f=0$ in Theorem \ref{thm:primal_convergence} and using 
  the concavity of $w(\cdot,\rho)$ for every $\rho\in(-1,1)$, we 
  have that $w(\cdot,\rho)\longrightarrow w(\cdot,0)$ uniformly on compacts 
  in $(-\phimin,\infty)$ as $\rho\rightarrow 0$.  Thus,
  $$
    \lim_n w(x+p(x,\rho_n),\rho_n) 
    = w(x+\bar{p},0),
  $$
  which implies that $w(x+\bar{p},0)<w(x+p(x,0),0)$.  Since $w(\cdot,0)$
  is strictly increasing, we conclude that $\bar{p}<p(x,0)$.
\end{proof}

\subsection{Proof of the Main Result}
The proof of the main result, Theorem \ref{thm:main_result},  follows Lemmas \ref{lemma:primal_lsc} and \ref{lemma:dual_usc} (below), which establish lower and upper semicontinuity-type results for the sequence of primal and dual value functions, respectively.

\begin{lemma}\label{lemma:primal_lsc}
  Suppose that the sequence of markets satisfies Assumption \ref{ass:semimg},
  and $\sM^\infty_V\neq\emptyset$.
  Then for $x\in\R$ and $x_n\longrightarrow x$ as $n\rightarrow\infty$,
  $$
    u_\infty(x) \leq \liminf_{n\rightarrow\infty} u_n(x_n)\,.
  $$
\end{lemma}

Significant difficulty in proving Lemma \ref{lemma:primal_lsc} stems from the nonequivalence of markets (the martingale drivers, $M^n$, differ). 
The idea behind the proof of Lemma \ref{lemma:primal_lsc} is that since the pre-limiting markets are ``close'' to the $S^\infty$ market, strategies in the $S^\infty$ market are ``close'' to being strategies in the pre-limiting markets.  This idea will be made precise by appropriate approximation and stopping.  First, we need a helper lemma.

\begin{lemma}\label{lemma:S_semimg_convergence}
Under Assumption \ref{ass:semimg}, $S^n\longrightarrow S^\infty$ in the semimartingale topology as $n\rightarrow\infty$.
\end{lemma}
\begin{proof}
Since $M^n\longrightarrow M^\infty$ in the semimartingale topology as $n\rightarrow\infty$, it remains to show that $\left(\lambda^n\cdot\left<M^n\right>\right)\longrightarrow\left(\lambda^\infty\cdot\left<M^\infty\right>\right)$ in the semimartingale topology.  We seek to show
\begin{equation}\label{eqn:variation_converging}
  A_n:=\sum_{i=1}^d \int_0^T \left|\lambda^n(\sigma^{n,i})^2-\lambda^\infty(\sigma^{\infty,i})^2\right| dt \longrightarrow 0 \ \text{ in probability as } n\rightarrow\infty,
\end{equation}
which will then imply the desired result.

The mapping $X\mapsto \left<X\right>$ is continuous in the space of semimartingales with respect to semimartingale convergence, and so Assumption \ref{ass:semimg} implies:
  \begin{align}
    \label{enum:Ms} \sum_{i=1}^d\int_0^T 
      \left(\sigma^{n,i}-\sigma^{\infty,i}\right)^2dt &\longrightarrow 0
       \ \text{ in probability as  } n\rightarrow\infty,\\
    \label{enum:lambdaMs} \sum_{i=1}^d\int_0^T
      \left(\lambda^n\sigma^{n,i}-\lambda^\infty\sigma^{\infty,i}\right)^2dt 
      &\longrightarrow 0 \ \text{ in probability as } n\rightarrow\infty.
  \end{align}
  
  Let $\{A_n\}_{n\in N}$ be a subsequence of $\{A_n\}_{n\in\bN}$, 
  where for notational convenience
  we denote the subsequence index $N$ as an infinite subset of $\bN$.  
  We choose a
  further subsequence $\{A_n\}_{n\in N'}$, where $N'\subset N$, such that 
  the convergence in \eqref{enum:Ms} and \eqref{enum:lambdaMs} occurs
  $\bP$-a.s. as $n\rightarrow\infty$ for $n\in N'$.  
    
  We define the random variable 
  $$
    q:=\sup_{n\in N'} \sum_{i=1}^d\int_0^T \left(\left(\sigma^{n,i}\right)^2 + \left(\lambda^n\sigma^{n,i}\right)^2\right)dt.
  $$
  The almost-sure convergence along the subsequence $N'$ implies that
  $q<\infty$, $\bP$-a.s., which allows us to define the equivalent probability
  measure $\frac{d\bQ}{d\bP}:= \frac{e^{-q}}{\E^\bP\left[e^{-q}\right]}$.
  Under $\bQ$, we have more regularity of elements in $N'$; in particular,
  \begin{equation}\label{eqn:convergence_in_measure}
    \sum_{i=1}^d \left(\sigma^{n,i}-\sigma^{\infty,i}\right)^2
    + \left(\lambda^n\sigma^{n,i}-\lambda^\infty\sigma^{\infty,i}\right)^2
    \longrightarrow 0 \ \text{ in $L^1(\bQ\times\Leb)$ 
    as $N'\ni n\rightarrow\infty$,}
  \end{equation}
  where $\Leb$ denotes the Lebesgue measure on $[0,T]$.  
  Hence, $\{\sum_{i=1}^d (1+(\lambda^n)^2)(\sigma^{n,i})^2\}_{n\in N'}$ 
  is $(\bQ\times\Leb)$-uniformly integrable.  Since $\sum_{i=1}^d |\lambda^n(\sigma^{n,i})^2-\lambda^\infty(\sigma^{\infty,i})^2|\leq \sum_{i=1}^d [(1+(\lambda^n)^2)(\sigma^{n,i})^2+(1+(\lambda^\infty)^2)(\sigma^{\infty,i})^2]$ for all $n\geq 1$ and by \eqref{eqn:convergence_in_measure}, $\sum_{i=1}^d |\lambda^n(\sigma^{n,i})^2-\lambda^\infty(\sigma^{\infty,i})^2|\longrightarrow 0$ in $(\bQ\times\Leb)$-measure as $n\rightarrow\infty$, we have that
  \begin{equation*}\label{eqn:L1Q_convergence}
    \sum_{i=1}^d \left|\lambda^n(\sigma^{n,i})^2-
    \lambda^\infty(\sigma^{\infty,i})^2\right|
    \longrightarrow 0 \ \text{ in $L^1(\bQ\times\Leb)$ as } 
    N'\ni n\rightarrow\infty.
  \end{equation*}
  
  Now we choose a further subsequence $\{A_n\}_{n\in N''}$, where 
  $N''\subseteq N'$, such that
  \begin{equation*}
    \sum_{i=1}^d \int_0^T \left|\lambda^n(\sigma^{n,i})^2 - 
    \lambda^\infty(\sigma^{\infty,i})^2\right|
    \longrightarrow 0 \ \text{ $\bQ$-a.s. as $N''\ni n\rightarrow\infty$},
  \end{equation*}
  and note that this convergence also holds $\bP$-a.s. by the equivalence
  of $\bP$ and $\bQ$.  Thus, we have shown 
  that for all subsequences $\{A_n\}_{n\in N}$, $N\subseteq \bN$, 
  there exists a further
  subsequence $\{A_n\}_{n\in N''}$, $N''\subseteq N$, 
  such that $\sum_{i=1}^d \int_0^T 
  |\lambda^n(\sigma^{n,i})^2 - \lambda^\infty(\sigma^{\infty,i})^2|
  \longrightarrow 0$ $\bP$-a.s. for $n\in N''$ as $n\rightarrow\infty$.
  Therefore, \eqref{eqn:variation_converging} holds, which completes the
  proof.
  
\end{proof}

\begin{proof}[Proof of Lemma \ref{lemma:primal_lsc}]
  First, we show that the supremum in the limiting primal optimization problem,
  \eqref{def:primal_value_fn}, can be taken over all admissible wealth
  processes whose integrands are bounded.  
  Let $H\in\adm^\infty$ be given, and let
  $K\in (0,\infty)$ be such that $(H\cdot S^\infty)_t\geq -K$ for all 
  $t\in[0,T]$.  Lemma \ref{lemma:H_approx1} provides us with 
  a sequence of integrands $\{H^n\}_{n\geq 1}$ such that for each $n\geq 1$, 
  $H^n$ is bounded uniformly in $t$ and $\omega$ while 
  $(H^n\cdot S^\infty)_t\geq -2K$ for all $t\in[0,T]$ and 
  $(H^n\cdot S^\infty)_T\longrightarrow (H\cdot S^\infty)_T$ 
  in probability as $n\rightarrow\infty$.  In particular, 
  $(H^n\cdot S^\infty)\in\adm^\infty$ with 
  $\{(H^n\cdot S^\infty)\}_{n\geq 1}$ sharing the same lower 
  admissibility bound, $-2K$.  By Fatou's Lemma,
  $$
    \E\left[U\left(x+(H\cdot S^\infty)_T+f\right)\right]
    \leq\liminf_n\rightarrow\infty \E\left[U\left(x+(H^n\cdot S^\infty)_T
    + f\right)\right].
  $$
  Therefore, it suffices to take the supremum in \eqref{def:primal_value_fn}
  over all $\tilde{H}\in\adm^\infty$ such that $\tilde{H}$ is uniformly bounded in 
  $t$ and $\omega$.
  That is,
  \begin{equation}\label{eqn:primal_sup}
    u_\infty(x) = \sup_{\tilde{H}\in\adm^\infty, 
    \tilde{H}\text{ bdd}} 
    \E\left[U\left(x+(\tilde{H}\cdot S^\infty)_T+f\right)\right].
  \end{equation}
  
  Now let $\tilde{H}\in\adm^\infty$ be given such that $\tilde{H}$ is uniformly bounded in $t$ 
  and $\omega$ by a constant $K\in(0,\infty)$.  
  Even though $\tilde{H}$ is $S^\infty$-admissible and 
  $S^n$-integrable for every $n$,
  it is not necessarily admissible 
  for each $S^n$ market.  Using Lemma \ref{lemma:H_approx2}, we 
  mitigate this issue by choosing $\{\tilde{H}^n\}_{n\geq 1}$ such that for 
  each $n\geq 1$, $\tilde{H}^n$ is bounded uniformly in $t$ and $\omega$ while 
  $(\tilde{H}^n\cdot S^n)_t\geq-3K$ for all $t\in[0,T]$ and $(\tilde{H}^n\cdot 
  S^n)_T\longrightarrow (\tilde{H}\cdot S^\infty)_T$ in probability as $n
  \rightarrow\infty$.
  
Applying Fatou's Lemma gives us that
  \begin{align*}
    \E\left[U(x+(\tilde{H}\cdot S^\infty)_T+f)\right]
    &\leq \liminf_{n\rightarrow\infty} 
      \E\left[U(x_n+(\tilde{H}^n\cdot S^n)_T+ f)\right] \\
    &\leq \liminf_{n\rightarrow\infty} u_n(x_n).
  \end{align*}
  Taking the supremum over all uniformly bounded 
  $\tilde{H}\in\adm^\infty$, as in
  \eqref{eqn:primal_sup}, yields the result.
\end{proof}

We next proceed to the second main lemma, which establishes an upper-semicontinuity result for the dual problem.
\begin{lemma}\label{lemma:dual_usc}
  Let the assumptions of the model be as in Theorem \ref{thm:main_result}.
  Then for $\{y_n\}_{1\leq n<\infty}\subseteq (0,\infty)$ such that 
  $y_n\longrightarrow y>0$ as $n\rightarrow\infty$,
  $$
    v_\infty(y) \geq \limsup_{n\rightarrow\infty} v_n(y_n)\,.
  $$
\end{lemma}

Using Assumption \ref{ass:bdd_inf}, the following lemma will further refine the collection $\sB$ over which the infimum is taken in the limiting market's dual problem.  
We define $\sB'$ by
\begin{equation}\label{eqn:Bp}
\begin{split}
  \sB' := \left\{L\in\sB: \right.&\exists \text{ constants } c=c(L),d=d(L), \\
  &\left. 0<c\leq\sE(L)_t\leq d<\infty, \,\forall t\in[0,T], \text{ and } 
  \left<L\right>_T\leq d \right\}
\end{split}
\end{equation}
The following lemma builds on Corollary 3.4 in \cite{LZ07SPA}.  
\begin{lemma}\label{lemma:bdd_L}
  Suppose that the limiting market's dual problem satisfies 
  Assumption \ref{ass:bdd_inf} and that $\E[V(Z^\infty_T)]<\infty$, 
  where $Z^\infty_T$ is the minimal martingale density for $S^\infty$.
  Let $\sB'$ be defined as in \eqref{eqn:Bp}.  Then for $y>0$,
  $$
    v_\infty(y) = \inf_{L\in\sB'} \E\left[V\left(y Z^\infty_T\sE(L)_T\right)
      + y Z^\infty_T \sE(L)_T f\right]\,.
  $$
\end{lemma}

\begin{proof}
  The first part of the proof is based on the proof of Corollary 3.4
  of \cite{LZ07SPA}.  Let $L\in\sB$ be given.
  By the convexity of $V$, we have
  \begin{align*}
    \E&\left[V\left(yZ^\infty_T\left(\frac{1}{n}+\frac{n-1}{n}\sE(L)_T\right)\right) + yZ^\infty_T\left(\frac{1}{n}+\frac{n-1}{n}\sE(L)_T\right)f\right] \\
    &\leq \frac{1}{n}\E\left[V\left(yZ^\infty_T\right)+yZ^\infty_T f\right]
      + \frac{n-1}{n}\E\left[V\left(yZ^\infty_T\sE(L)_T\right)
      + yZ^\infty_T \sE(L)_T f\right] \\
    &\longrightarrow \E\left[V\left(yZ^\infty_T\sE(L)_T\right)
      + yZ^\infty_T\sE(L)_T f\right] \ \ \text{ as $n\rightarrow\infty$},
  \end{align*}
  because $V(yZ^\infty_T)\in L^1(\bP)$ by the
  assumption that $\E[V(Z^\infty_T)]<\infty$ and reasonable asymptotic
  elasticity, \eqref{eqn:ae}.
  For each $n\geq 1$, we let $L^n$ denote the element $L^n\in\sB$ such that
  $\frac{1}{n}+\frac{n-1}{n}\sE(L) = \sE(L^n)$.

  Let $\eps>0$ be given, and choose $N$ sufficiently large such that
  $$
    \E\left[V\left(yZ^\infty_T\sE(L^N)_T\right)
      + yZ^\infty_T\sE(L^N)_T f\right]\leq 
    \E\left[V\left(yZ^\infty_T\sE(L)_T\right)
      + yZ^\infty_T\sE(L)_T f\right] + \eps.
  $$
  
  We define the sequence of stopping times $\{\tau_k\}_{1\leq k<\infty}$ by
  $\tau_k:=\inf\left\{t\leq T: \left<L^N\right>_t\geq k\right\}$.
  Then $(L^N)^{\tau_k}\in\sB'$ for each $k$. 
  By continuity of $L^N$ and
  finiteness of $\left<L^N\right>_T$, we have that $\sE(L^N)_{\tau_k}
  \longrightarrow \sE(L^N)_T$ in probability as $k\rightarrow\infty$.  
  Scheffe's Lemma implies that the $L^1(\bP)-\lim_k Z^\infty_T\sE(L^N)_{\tau_k} = Z^\infty_T\sE(L^N)_T$, which implies that  $\lim_k
  \E\left[y Z^\infty_T\sE(L^N)_{\tau_k}\, f\right] 
  = \E\left[y Z^\infty_T\sE(L^N)_T f\right]$.
  
  Convergence in probability of 
  $\left\{\sE(L^N)_{\tau_k}\right\}_{1\leq k<\infty}$ also
  implies that $V(yZ^\infty_T\sE(L^N)_{\tau_k})
  \longrightarrow V(y Z^\infty_T\sE(L^N)_T)$ 
  in probability as $k\rightarrow\infty$.
  Let $C$ be the bound on $\sE(L^N)$ from above given to us in
  definition of $\sB$.  Since $\frac{1}{N}\leq \sE(L^N)_t\leq C$ for all
  $t$, 
  we have for all $k$ that $V(yZ^\infty_T\sE(L^N)_{\tau_k})\leq
  \max\left(V(\frac{1}{N}Z^\infty_T), V(CZ^\infty_T)\right)$, where 
  $\max\left(V(\frac{1}{N}Z^\infty_T), V(CZ^\infty_T)\right)$ is in $L^1(\bP)$
  by reasonable asymptotic elasticity, \eqref{eqn:ae}.  Thus, 
  $V(yZ^\infty_T\sE(L^N)_{\tau_k}) \longrightarrow V(y Z^\infty_T\sE(L^N)_T)$ 
  in $L^1(\bP)$ as $k\rightarrow\infty$.
  
  We may choose $K$ sufficiently large so that 
  $\E\left[V\left(yZ^\infty_T\sE(L^N)_{\tau_K}\right)+yZ^\infty_T\sE(L^N)_{\tau_K}f\right] 
  \leq \E\left[V\left(yZ^\infty_T\sE(L^N)_T\right)+yZ^\infty_T\sE(L^N)_T f\right]+\eps$,  
  which then implies 
  that
  $$
    \E\left[V\left(yZ^\infty_T\sE(L^N)_{\tau_K}\right)
    +y Z^\infty_T\sE(L^N)_{\tau_K}\, f\right]
    \leq \E\left[V\left(yZ^\infty_T\sE(L)_T\right)
    +y Z^\infty_T\sE(L)_T\, f\right] + 2\eps.
  $$
  Since $\eps>0$ and $L\in\sB$ are arbitrary, Assumption \ref{ass:bdd_inf}
  allows us to conclude the desired result.
\end{proof}

Establishing an upper-semicontinuity property for the dual problem is difficult because with small changes in the limiting market, we must produce a dual element of a pre-limiting market with appropriately small changes.  Using the Kunita-Watanabe decomposition, we decompose elements $L\in\sB'$ in terms of strongly orthogonal components based on the varying martingale drivers, $M^n$.  See \cite{KW67} for more general coverage of the Kunita-Watanabe decomposition.

A $\bP$-local martingale, $N$, is said to be in $H^2_0(\bP)$ provided $N_0=0$ and $\E[\left<N\right>_T]<\infty$, in which case $N$ is a martingale.  A sequence of martingales $\{N^n\}_{1\leq n<\infty}\subseteq H^2_0(\bP)$ converges to $N$ in $H^2_0(\bP)$ if $\E[\left<N^n-N\right>_T]\longrightarrow 0$ as $n\rightarrow\infty$.  We say that two local martingales, $M$ and $N$, are \textit{strongly orthogonal} if $\left<M,N\right>_t = 0$ for all $t\in[0,T]$.
\begin{lemma}\label{lemma:L_approximation}
  Let $\left\{M^n\right\}_{1\leq n\leq\infty}$ be local 
  martingales such that $M^n\longrightarrow M^\infty$ in the
  semimartingale topology as $n\rightarrow\infty$, 
  and suppose that $M^\infty$ satisfies Assumption \ref{ass:nondegeneracy}.  
  Let $L\in H^2_0(\bP)$ be strongly orthogonal
  to $M^\infty$,  and 
  for $1\leq n<\infty$, decompose $L$ into its
  (unique) Kunita-Watanabe decomposition,
  $$
    L=L^n+(H^n\cdot M^n),
  $$
  where $L^n$ and $(H^n\cdot M^n)$ are in $H^2_0(\bP)$ and 
  $L^n$ is strongly orthogonal to $M^n$.
  Then $L^n\longrightarrow L$ in $H^2_0(\bP)$ 
  as $n\rightarrow\infty$.
\end{lemma}
\begin{proof} 
  The filtration $\bF=(\sF_t)_{0\leq t\leq T}$ is the ($\bP$-completed)
  filtration generated by the $d$-dimensional Brownian motion
  $(B^1,\ldots,B^d)$ on $(\Omega,\sF,\bF,\bP)$ with $\sF=\sF_T$.
  For notational concreteness, we denote
  $$
    L = (\nu^1\cdot B^1)+\ldots+(\nu^d\cdot B^d), 
  $$
  for $\nu^k\in\sL^2$, $1\leq k\leq d$.  
  For $\x=(x_1,\ldots,x_d),\y=(y_1,\ldots,y_d)\in\R^d$, 
  we let $|\x|$ denote the Euclidean norm, 
  $|\x|:=\sqrt{x_1^2+\cdots+ x_d^2}$, and let the inner product be 
  $\x\LargerCdot\y:=x_1y_1+\ldots+x_dy_d$.  
  We define the vector $\nu := (\nu^1,\ldots,\nu^d)$.
  
  For $1\leq n<\infty$, we define
  $$
    H^n := \frac{\nu\LargerCdot\sigma^n}{|\sigma^n|^2}
    \, \bI_{\{|\sigma^n|\neq 0\}}.
  $$
  Then $H^n$ is progressively measurable and $M^n$-integrable with 
  $(H^n\cdot M^n)\in H^2_0(\bP)$. 
  We define $L^n:= L-(H^n\cdot M^n)\in H^2_0(\bP)$. 
  $L^n$ and $M^n$ are strongly orthogonal,  
  and thus $L=L^n+(H^n\cdot M^n)$ is the 
  Kunita-Watanabe decomposition for $L$ with respect to $M^n$. 
  Since $L^n$ and $M^n$ are strongly orthogonal,
  $L^n\longrightarrow L$ in $H^2_0(\bP)$ if and only if 
  $(H^n\cdot M^n)\longrightarrow 0$ in $H^2_0(\bP)$ as $n\rightarrow\infty$.  
  Hence, we seek to show that
  $\E[\left<H^n\cdot M^n\right>_T]
  = \E\left[\int_0^T \frac{(\nu\LargerCdot\sigma^n)^2}{|\sigma^n|^2}
  \bI_{\{|\sigma^n|\neq 0\}}dt\right]\longrightarrow 0$ as 
  $n\rightarrow\infty$. 
  
  Since $L\in H^2_0(\bP)$, we have
  for $1\leq n<\infty$,
  $$
    \frac{\left(\nu\LargerCdot\sigma^n\right)^2}
    {|\sigma^n|^2}\,\bI_{\{|\sigma^n|\neq 0\}}
    \leq |\nu|^2
    \in L^1(\bP\times\Leb).
  $$
  The assumption that $M^n\longrightarrow M^\infty$ in the
  semimartingale topology as 
  $n\rightarrow\infty$ implies that for $1\leq k\leq d$, $\sigma^{n,k}\longrightarrow\sigma^{\infty,k}$ in $(\bP\times\Leb)$-measure as $n\rightarrow\infty$.  Since $\left<L,M^\infty\right>=0$,
  we have that $\nu\LargerCdot\sigma^\infty=0$ $(\bP\times\Leb)$-a.e.  
  Assumption \ref{ass:nondegeneracy} ensures that
  $|\sigma^\infty|\neq 0$ ($\bP\times\Leb$)-a.e., and hence,
  $$
    \frac{\left(\nu\LargerCdot\sigma^n\right)^2}{|\sigma^n|^2}
    \, \bI_{\{|\sigma^n|\neq 0\}} \longrightarrow 0
    \ \ \text{ in ($\bP\times\Leb$)-measure as $n\rightarrow\infty$.}
  $$
  Thus dominated convergence implies that
  $\E[\left<H^n\cdot M^n\right>_T]\longrightarrow 0$ as $n\rightarrow\infty$,
  which completes the proof of the claim.
\end{proof}

\begin{proof}[Proof of Lemma \ref{lemma:dual_usc}]
  We let $\sB'$ be defined as in \eqref{eqn:Bp} and let $L\in\sB'$ be given.  Let $K\in(0,\infty)$ be the constant given in the definition of $\sB'$ such that $|L_t|\leq K$ for all $t$ and $\left<L\right>_T\leq K$.
  
  We let $L^n$ be given as in Lemma \ref{lemma:L_approximation}.  Then $L^n\longrightarrow L$ in $H^2_0$ as $n\rightarrow\infty$.  For $1\leq n<\infty$, define stopping times $\tau_n:=\inf\{t\leq T: |L^n_t-L_t|\geq 1 \text{ or } \left<L^n\right>_t\geq K+1\}$.  The $H^2_0(\bP)$ convergence of $\{L^n\}_{1\leq n<\infty}$ implies that $\left<L^n\right>_T\longrightarrow\left<L\right>_T$ in $L^1(\bP)$ as $n\rightarrow\infty$, while the Burkholder-Davis-Gundy inequalities additionally give us that $\bP(\sup_t |L^n_t-L_t|\geq 1)\longrightarrow 0$ as $n\rightarrow\infty$.  Hence, $\bP(\tau_n=T)\longrightarrow 1$ as $n\rightarrow\infty$.  We conclude that $L^n_{\tau_n}\longrightarrow L_T$ and $\left<L^n\right>_{\tau_n}\longrightarrow\left<L\right>_T$ in probability as $n\rightarrow\infty$, which yields
  $$
    \sE(L^n)_{\tau_n}\longrightarrow \sE(L)_T
    \ \ \text{ in probability as $n\rightarrow\infty$.}
  $$
  Furthermore, the definition of $\tau_n$ provides upper and lower bounds 
  on $\sE(L^n)_{\tau_n}$, which are independent of $n$:
  \begin{equation}\label{eqn:exp_bound}
    e^{-2K-2}\leq \sE(L^n)_{\tau_n}\leq e^{K+1}.
  \end{equation}
  Such uniform bounds and the choice of the $L^n$s are made possible by the choice of $L\in\sB'$.

For $1\leq n\leq\infty$, $Z^n$ is a martingale by Assumption \ref{ass:Z_mg}, and by Fatou's Lemma, $1 = \E[Z^\infty_T] \leq \liminf_{n\rightarrow\infty} \E[Z^n_T] = 1$.  Hence, $\lim_{n\rightarrow\infty} \E[Z^n_T] = \E[Z^\infty_T]$.  Scheffe's Lemma then implies that $Z^n_T\longrightarrow Z^\infty_T$ in $L^1(\bP)$ as $n\rightarrow\infty$, and in particular, $\{Z^n_T\}_{1\leq n\leq \infty}$ is uniformly integrable.  By \eqref{eqn:exp_bound} and since $f\in L^\infty(\bP)$, we have that
$$
  0 \leq y_n Z^n_T \sE(L^n)_{\tau_n} f
  \leq \left(\sup_m y_m\right) e^{K+1} \|f\|_\infty Z^n_T,
$$
which implies that $\{y_nZ^n_T\sE(L^n)_{\tau_n} f\}_{1\leq n\leq\infty}$ is uniformly integrable.  Convergence in probability and uniform integrability imply that
$$
  y_nZ^n_T\sE(L^n)_{\tau_n}f \longrightarrow y Z^\infty_T\sE(L)_Tf
  \ \text{ in $L^1(\bP)$ as $n\rightarrow\infty$}.
$$

We use \eqref{eqn:exp_bound} again in order to obtain uniform integrability of the remaining term in the dual value function.  As mentioned in Assumption 1.2(i) of \cite{OZ09MF}, the reasonable asymptotic elasticity condition \eqref{eqn:ae} along with the $U(0)>0$ is equivalent to the following:  for all $\lambda>0$ there exists $C>0$ such that $V(\lambda y)\leq C V(y)$ for all $y\geq 0$.  Then for $1\leq n<\infty$,
  \begin{align*}
    0 &\leq V\left(y_nZ^n_T\sE(L^n)_{\tau_n}\right) \\
      &\leq V\left(y_nZ^n_T e^{K+1}\right)
        \bI_{\{y_n Z^n_T\sE(L^n)_{\tau_n}\geq U'(0)\}}
        + V\left(y_n Z^n_T e^{-2K-2}\right)
        \bI_{\{y_n Z^n_T\sE(L^n)_{\tau_n}< U'(0)\}} \\
      &\leq V\left(\left(\sup_m y_m\right) e^{K+1}Z^n_T\right)
        + V\left(\left(\inf_m y_m\right)e^{-2K-2}Z^n_T\right) \\
      &\leq (C_1+C_2)V(Z^n_T),
  \end{align*}
  where $C_1, C_2$ are the constants produced by the reasonable
  asymptotic elasticity of $U$.  The constants
  $C_1, C_2$ depend on the choice of $L$, $K$, 
  $\inf_m y_m$, and $\sup_m y_m$ but not on $n$.  
  Assumption \ref{ass:V_ui} now guarantees the uniform integrability of
  $\left\{V\left(y_n Z^n_T\sE(L^n)_{\tau_n}\right)\right\}_{1\leq n<\infty}$.   
  Convergence in probability and uniform integrability imply that
  $V\left(y_n Z^n_T\sE(L^n)_{\tau_n}\right)\longrightarrow 
  V\left(y Z^\infty_T\sE(L)_T\right)$ in $L^1(\bP)$ as $n\rightarrow\infty$. 
  Finally, we have
  $$
    \E\left[V(yZ^\infty_T\sE(L)_T)+yZ^\infty_T\sE(L)_Tf\right] 
    = \lim_n \E\left[V(y_n Z^n_T\sE(L^n)_{\tau_n})
      + y_n Z^n_T \sE(L^n)_{\tau_n}f\right] 
    \geq \limsup_n v_n(y_n).
  $$
  Taking the infimum over all $L\in\sB'$ and applying Lemma \ref{lemma:bdd_L}  
  yields $v_\infty(y)\geq \limsup_n v_n(y_n)$.
\end{proof}

\begin{proof}[Proof of Theorem \ref{thm:main_result}]\label{proof:main_result}
  We first note that the assumption that
  $\sM^\infty_V\neq\emptyset$ of Lemma \ref{lemma:primal_lsc} is satisfied by 
  Assumption \ref{ass:V_ui}.  For $x_n\longrightarrow x\in\R$
  and $y=y(x)$, Lemmas \ref{lemma:primal_lsc} and \ref{lemma:dual_usc} imply
  \begin{equation}\label{eqn:inequality_chain}
    u_\infty(x)
      \leq \liminf_{n\rightarrow\infty} u_n(x_n)
      \leq \limsup_{n\rightarrow\infty} v_n(y)+x_ny
      \leq v_\infty(y)+xy = u_\infty(x).
  \end{equation}
  The last equality can be shown by Theorem 1.1 of \cite{OZ09MF} by taking $\sE = x+f$ and $y=\E\left[\frac{d\hat{\mu}(x)}{d\bP}\right]$.  Here, $\sE$ and $\hat{\mu}(x)$ refer to the notation used in \cite{OZ09MF}.
  
  Moreover, the inequality chain \eqref{eqn:inequality_chain} shows 
  that for $y>0$, 
  $v_n(y)\longrightarrow v_\infty(y)$ as $n\rightarrow\infty$.  
  For $y_n\longrightarrow y>0$, we also have that $v_n(y_n)\longrightarrow v_\infty(y)$ as $n\rightarrow\infty$ because the convexity of each $v_n$ implies 
  that $v_n\longrightarrow v_\infty$ uniformly on compacts in $(0,\infty)$ 
  as $n\rightarrow\infty$.
\end{proof}

\begin{proof}[Proof of Corollary \ref{cor:indifference_price}]
  Let $\{p_{n_k}(x)\}_{1\leq k<\infty}$ be a convergent subsequence of 
  $\{p_n(x)\}_{1\leq n<\infty}$ 
  with $\lim_k p_{n_k}(x) = p\in\R$.  By Theorem \ref{thm:main_result},
  $$
    u_\infty(x) = \lim_k u_{n_k}(x),
  $$
  while $w_{n_k}(x+p_{n_k}(x)) = u_{n_k}(x)$ for each $k\geq 1$ by the
  definition of the indifference price.  Next, we take the contingent
  claim to be $0$ and note
  that $\lim_k x+p_{n_k}(x) = x+p$, which allows us to conclude from Theorem
  \ref{thm:main_result} that
  $$
    w_\infty(x+p) = \lim_k w_{n_k}(x+p_{n_k}(x)),
  $$
  which implies that $p=p_\infty(x)$.  Since $f\in L^\infty(\bP)$,
  $\{p_n(x)\}_n$ is bounded, hence any subsequence has a further subsequence
  that converges to $p_\infty(x)$.  Therefore, $\lim_n p_n(x)$ exists and 
  equals $p_\infty(x)$.
\end{proof}

\section{Examples}
\label{section:examples}
The first example shows that Assumption \ref{ass:nondegeneracy} is necessary in the sense that its absence can allow Theorem \ref{thm:main_result}'s conclusion to fail.  The example is constructed in a very simple setting, but the same idea can generate more complex counterexamples whenever Assumption \ref{ass:nondegeneracy} fails.
\begin{example}\label{example:degeneracy_ass}
  Let $d=1$, so that the probability space is generated by a $1$-dimensional
  Brownian motion, $B$.  We define the martingales $M^n:=\frac{1}{n} B$ 
  for $1\leq n<\infty$ and $M^\infty:=0$.  Let $\lambda^n := 0$ for all $1\leq n\leq\infty$ so that $S^\infty_t = 0$ for all $t\in[0,T]$ and for $1\leq n<\infty$, $S^n$ has the dynamics
  $$
    dS^n = \frac{1}{n}dB,\ \ \ S^n_0=0.
  $$
  The stock markets satisfy Assumptions \ref{ass:semimg} 
  and \ref{ass:Z_mg}, but the limiting market
  does not satisfy Assumption \ref{ass:nondegeneracy}.
  
  Let the contingent claim be given by $f := \bI_{\{B_T\geq 0\}}$.  By 
  It\^o's representation theorem and the boundedness of $f$, 
  there exists a progressively measurable $H$ such that 
  $\E[\int_0^T H^2_t dt]<\infty$ and 
  $f = \frac{1}{2}+(H\cdot B)_T$.  Moreover, $(H\cdot B)$ is bounded 
  since for all $t\in[0,T]$,
  $$
    (H\cdot B)_t = \E\left[(H\cdot B)_T \left| \right. \sF_t\right]
    = \E\left[f-\frac{1}{2}\left|\right.\sF_t\right]
    \in \left[-\frac{1}{2},\frac{1}{2}\right], \ \text{ $\bP$-a.s.}.
  $$
  Hence, we can conclude by
  Theorem 2.1 of \cite{S01AAP} that for all $1\leq n<\infty$ and $x\in\R$,
  $$
    u_n(x) = U\left(x+\frac{1}{2}\right).
  $$
  Yet for all $x\in\R$, Jensen's inequality implies that $u_\infty(x) = \E\left[U\left(x+f\right)\right] < U\left(x+\frac{1}{2}\right)$.
\end{example}
The following two examples provide sufficient conditions on the limiting market for Assumption \ref{ass:bdd_inf} to hold.  

\begin{example}\label{ex:complete} 
This example covers the original motivation for this work, where the contingent claim is replicable in the (possibly incomplete) limiting market.  In this case, the limiting market consists of a driving Brownian motion, a replicable claim, and additional independent Brownian noise.

Recall that $(B^1,\ldots,B^d)$ is the $d$-dimensional Brownian motion generating the completed filtration, $\bF=(\sF_t)_{0\leq t\leq T}$.  Let $(\sF^1_t)_{0\leq t\leq T}$ denote the filtration generated by $B^1$, completed with all $\bP$-null sets. 
The risky asset, $S^\infty$, has dynamics as in \eqref{eqn:S_dynamics} and is $(\sF^1_t)_{0\leq t\leq T}$-adapted.  The contingent claim, $f\in L^\infty(\Omega,\sF^1_T,\bP)$, is replicable:  there exists an $S^\infty$-integrable $H$ and constant $c$ such that $f = c+(H\cdot S^\infty)_T$.

\begin{proposition}
  Suppose that $S^\infty$ is $\left(\sF^1_t\right)_{0\leq t\leq T}$-adapted
  with dynamics \eqref{eqn:S_dynamics} and satisfies Assumption
  \ref{ass:nondegeneracy}.  Suppose that 
  $f\in L^\infty(\Omega,\sF^1_T,\bP)$ is replicable.  
  Then Assumption \ref{ass:bdd_inf} is satisfied.
\end{proposition}

\begin{proof}
Let $y>0$ and $\bQ\in\sM^\infty_V$ be given.  Write $\frac{d\bQ}{d\bP}=Z^\infty_T\sE(L)_T$ for its Radon-Nikodym density.  We have that $Z^\infty_T\in\sF^1_T$, while Assumption \ref{ass:nondegeneracy} implies that $\left<L,B^1\right>_t=0$ for $t\in[0,T]$.
Note that $\E[\sE(L)_T|\sF^1_T] = 1$, $\bP$-a.s., since
$$
  1 = \E[Z^\infty_T\sE(L)_T] 
    = \E[Z^\infty_T\underbrace{\E[\sE(L)_T|\sF^1_T]}_{\leq \E[\sE(L)_0|\sF^1_T]\, =\, 1}\,] 
    \leq \E[Z^\infty_T] = 1,
$$
with equality holding if and only if $\E[\sE(L)_T|\sF^1_T]=1$, $\bP$-a.s.
By Jensen's inequality,
\begin{align*}
  \E[V(yZ^\infty_T\sE(L)_T)]
  &= \E\left[\E\left[V(yz\sE(L)_T)|\sF^1_T\right]|_{z=Z^\infty_T}\right]
  \\
  &\geq \E\left[V(yz\E\left[\sE(L)_T|\sF^1_T\right])|_{z=Z^\infty_T}\right]
  \\
  &= \E[V(yZ^\infty_T)].
\end{align*}
Since $f$ is bounded and replicable, $\bQ\mapsto \E[\frac{d\bQ}{d\bP}f]$ is constant on $\sM^\infty$.  Hence, for all $\bQ\in\sM^\infty_V$,
$
  \E[V(yZ^\infty_T)+yZ^\infty_T f] 
  \leq \E\left[V\left(y\frac{d\bQ}{d\bP}\right)
  + y\frac{d\bQ}{d\bP}f\right],
$
which implies that $Z^\infty_T$ is the density of the dual minimizer, and so Assumption \ref{ass:bdd_inf} is satisfied.
\end{proof}


\end{example}

\begin{example}[Exponential Investors]
For the exponential investor, Assumption \ref{ass:bdd_inf} is satisfied, under an easier-to-verify BMO assumption.  We refer to \cite{Kaz94} for additional details on $\bmo$ martingales. 

\begin{definition} \label{def:BMO}
  A $\bP$-local martingale $N$ 
  is said to be in $BMO(\bP)$ if
  $$
    \sup_{\tau} 
    \left\|\E^\bP\left[\left|N_T-N_\tau\right|\,
    |\sF_\tau\right]\right\|_\infty
    <\infty,
  $$
  where the supremum is taken over stopping times $\tau\leq T$.
\end{definition}

\begin{assumption}\label{ass:bmo}
  $(\lambda^\infty\cdot M^\infty)\in \bmo(\bP)$.
\end{assumption}

For the remainder of this section, we let $U(x)=-\exp(-\alpha x)$ for a positive constant $\alpha$.  The conjugate to $U$ is $V(y)=\frac{y}{\alpha}\left(\log \frac{y}{\alpha}-1\right)$, $y>0$.  We have the following relationships for $c\in\R$ and $y>0$:
\begin{gather}
  \label{formula:Vprime} V'(cy) = V'(y)+\frac{1}{\alpha}\log c, \\
  \label{formula:V} V(y)+yc = y\left(V'(ye^{\alpha c})-\frac{1}{\alpha}\right).
\end{gather}
For a set $A\in\sF$ and random variable $X\in L^1(\bP)$, we adopt the notation $\E[X; A] := \E[X\bI_A] = \int_A Xd\bP$.

\begin{theorem}\label{thm:exp_bmo_case}
  Let $U(x) = -\exp(-\alpha x)$ for a positive constant $\alpha$ and assume that 
  Assumption \ref{ass:bmo} holds.  Let $\Qmin$ denote the minimal martingale measure, $\frac{d\Qmin}{d\bP}:=Z^\infty_T=\sE(-\lambda^\infty\cdot M^\infty)_T$, and suppose that $\bQ^\infty\in\sM^\infty_V$.  Then Assumption
  \ref{ass:bdd_inf} is satisfied.
\end{theorem}
\begin{proof}
Let $x\in\R$ and $Z^\infty_T\sE(L)_T=\sE(-(\lambda^\infty\cdot M^\infty)+L)_T\in\sM^\infty_V$ be the dual optimizer for the dual problem \eqref{def:dual_value_fn} with $n=\infty$ and $y:=u_\infty'(x)$.
For $1\leq n<\infty$, we define the stopping times $\tau_n:=\inf\{t\leq T: \sE(L)_t\geq n\}$.  Using that $V(0)=0$ and the definition of $\tau_n$, it is not difficult to verify that each probability density $Z^\infty_T\sE(L)_{\tau_n}$ corresponds to a martingale measure in $\sM^\infty_V$.  

Theorem 2.1 of \cite{KS02MF} implies that there exists an $S^\infty$-integrable $\hat{H}$ such that $\hat{H}$ is optimal for \eqref{def:primal_value_fn} with $n=\infty$ and $(\hat{H}\cdot S^\infty)$ is a martingale with respect to every measure $\bQ\in\sM^\infty_V$.  The process $\hat{H}$ is a \textit{permissible} wealth process (in the $S^\infty$ market), rather than an admissible wealth process; see \cite{OZ09MF} Definition 1.1 for details.  
Proposition 4.1 from \cite{OZ09MF} implies that $x+(\hat{H}\cdot S^\infty)_T +f= -V'\left(yZ^\infty_T\sE(L)_T\right)$.  
Hence, for any $\bQ\in\sM^\infty_V$, \eqref{formula:Vprime} with $c=x$ implies that
\begin{equation}\label{eqn:0_expectation}
  \E\left[\frac{d\bQ}{d\bP}V'\left(yZ^\infty_T\sE(L)_T e^{\alpha f}\right)\right] 
  = \E\left[Z^\infty_T\sE(L)_T V'\left(yZ^\infty_T\sE(L)_T e^{\alpha f}\right)\right].
\end{equation} 
Then,
\begin{align*}
  0 &\leq \E\left[V(yZ^\infty_T\sE(L)_{\tau_n})
      +yZ^\infty_T\sE(L)_{\tau_n}f\right] - v_\infty(y) \\
    &= \E\left[V(yZ^\infty_T\sE(L)_{\tau_n})
      +yZ^\infty_T\sE(L)_{\tau_n}f\right] - 
      \E\left[V(yZ^\infty_T\sE(L)_T)+yZ^\infty_T\sE(L)_T f\right]\\
    &=\E\left[yZ^\infty_T\sE(L)_{\tau_n}V'(yZ^\infty_T\sE(L)_{\tau_n}e^{\alpha f})-yZ^\infty_T\sE(L)_TV'(yZ^\infty_T\sE(L)_Te^{\alpha f})\right]
        &\text{by \eqref{formula:V}}\\
    &=\E\left[yZ^\infty_T\sE(L)_{\tau_n}\left(V'(yZ^\infty_T\sE(L)_{\tau_n}
      e^{\alpha f})-V'(yZ^\infty_T\sE(L)_Te^{\alpha f})\right)\right]
        &\text{by \eqref{eqn:0_expectation}}\\
    &=\frac{y}{\alpha}\E\left[Z^\infty_T\sE(L)_{\tau_n}
      \left(\log\sE(L)_{\tau_n}-\log\sE(L)_T\right)\right]
        &\text{by \eqref{formula:Vprime}}\\
    &=\frac{y}{\alpha}\E^{\bQ^\infty}\left[n\log\left(\frac{n}{\sE(L)_T}
      \right);\{\tau_n<T\}\right]\\
    &=\frac{y}{\alpha}\left(n\log n\,\bQ^\infty(\tau_n<T)
      -n\,\E^{\bQ^\infty}[\log\sE(L)_T;\{\tau_n<T\}]\right).
\end{align*}

%
In order to show Assumption \ref{ass:bdd_inf}, it now suffices to show
\begin{equation}\label{eqn:to_show}
  n\log n\,\Qmin(\tau_n<T)-n\E^{\Qmin}[\log\sE(L)_T;\{\tau_n<T\}]
  \longrightarrow 0\ \ \text{ as $n\rightarrow\infty$}.
\end{equation}
Showing $n\log n\, \Qmin\left(\tau_n<T\right)\longrightarrow 0$ as $n\rightarrow\infty$ will employ Doob's submartingale inequality, whereas $n\, \E^\Qmin[\log\sE(L)_T; \{\tau_n<T\}]\longrightarrow 0$ relies on the assumption that $(\lambda^\infty\cdot M^\infty)\in \bmo(\bP)$.

Let $\phi(y):=y\log y$.  We have that $\phi$ is convex, $\phi\geq -1/e$, and $\phi$ is increasing on $[1/e,\infty)$.  Using that $Z^\infty_T\sE(L)_T$ is the dual optimizer, it is not difficult to check that $\phi(\sE(L)_t)\in L^1(\Qmin)$ for each $t\in[0,T]$.
Convexity of $\phi$ implies that $\phi(\sE(L))$ is a $\Qmin$-submartingale.  (Note that $\sE(L)$ is a $\Qmin$-martingale since $\E^\Qmin[\sE(L)_T] = \E^\bP[Z^\infty_T\sE(L)_T]=1$.)

For a process $Y$, we let $Y^*:=\sup_{0\leq t\leq T} Y_t$.  For any $n>1$,
$$
  \sE(L)^* \geq n \ \text{ if and only if }\ 
  \phi(\sE(L))^* = \left(\sE(L)\log\sE(L)\right)^*\geq n\log n.
$$
Doob's submartingale inequality implies that for $n>1$,
\begin{align*}
  n\log n\, \Qmin(\sE(L)^*\geq n)
  &= n\log n\, \Qmin\left(\phi(\sE(L))^*\geq n\log n\right) \\
  &\leq 
  \E^\Qmin\left[\phi(\sE(L)_T)^+; 
    \left\{\phi(\sE(L))^*\geq n\log n\right\}\right] \\
  &= 
  \E^\Qmin\left[\phi(\sE(L)_T)^+; 
    \left\{\sE(L)^*\geq n\right\}\right].
\end{align*}
Since $\phi(\sE(L)_T)\in L^1(\Qmin)$, we have that
\begin{align*}
  \limsup_{n\rightarrow\infty} \ n\log n\ \Qmin(\tau_n<T)
  &\leq \limsup_{n\rightarrow\infty} \ n\log n\ \Qmin(\sE(L)^*\geq n) \\
  &\leq \limsup_{n\rightarrow\infty} \ 
    \E^\Qmin[\phi(\sE(L)_T)^+;\{\sE(L)^*\geq n\}] \\
  &= 0.
\end{align*}

Now suppose that Assumption \ref{ass:bmo} holds.  Then by Lemma 3.1 of \cite{6AP} the density of the dual optimizer, $Z^\infty\sE(L)$, satisfies $\sR_{L \log L}(\bP)$; that is,  $Z^\infty\sE(L)$ is a $\bP$-martingale and
  $$
    \sup_{\tau} \left\| \E^\bP\left[\frac{Z^\infty_T\sE(L)_T}{Z^\infty_\tau\sE(L)_\tau}\left.
    \log\left(\frac{Z^\infty_T\sE(L)_T}{Z^\infty_\tau\sE(L)_\tau}\right)
    \right|\sF_\tau\right]\right\|_\infty<\infty,
  $$
  where the supremum is taken over all stopping times $\tau\leq T$.
Lemma 2.2 of \cite{GR02AP} shows that $-(\lambda^\infty\cdot M^\infty)+L\in\bmo(\bP)$, which then implies that $L\in\bmo(\bP)$.

Since $\left<-\lambda^\infty\cdot M^\infty, L\right>_t=0$ for all $t\in[0,T]$, then Theorem 3.6 of \cite{Kaz94} implies that $L = L-\left<-\lambda^\infty\cdot M^\infty,L\right> \in \bmo(\Qmin)$.  Then by Theorem 2.4 of \cite{Kaz94}, $L$ satisfies
\begin{equation}\label{eqn:logplus}
  \sup_\tau \left\|\E^\Qmin\left[\left.
  \log^+\left(\frac{\sE(L)_\tau}{\sE(L)_T}\right)
  \right|\sF_\tau\right]\right\|_\infty <\infty,
\end{equation}
where the supremum is taken over all stopping times $\tau\leq T$.
Re-writing \eqref{eqn:logplus}, and considering only the stopping times $\tau_n$ for $n\geq 1$, we have
$$
  K:= \sup_n\left\|\E^\Qmin\left[\left(\log\sE(L)_{\tau_n}-\log\sE(L)_T\right)
  \bI_{\left\{\sE(L)_{\tau_n}\geq\sE(L)_T\right\}}
  |\sF_{\tau_n}\right]\right\|_\infty < \infty.
$$
For each $n\geq 1$, $\{\tau_n<T\}\in\sF_{\tau_n}$ and $\sE(L)_{\tau_n}=n$ on $\{\tau_n<T\}$.  Then,
\begin{align*}
  -\E^\Qmin &\left[\log\sE(L)_T; \{\sE(L)_{\tau_n}\geq\sE(L)_T\}\cap\{\tau_n<T\}\right] \\
  &\leq \E^\Qmin \left[\log\sE(L)_{\tau_n}-\log\sE(L)_T; \{\sE(L)_{\tau_n}\geq\sE(L)_T\}\cap\{\tau_n<T\}\right] \\
  &= \E^\Qmin\left[\E^\Qmin\left[\left(\log\sE(L)_{\tau_n}-\log\sE(L)_T\right)\bI_{\{\sE(L)_{\tau_n}\geq\sE(L)_T\}}|\sF_{\tau_n}\right];
  \left\{\tau_n<T\right\}\right] \\
  &\leq K\, \Qmin\left(\tau_n<T\right).
\end{align*}
Thus,
\begin{align*}
  -n \E^\Qmin&\left[\log\sE(L)_T;\{\tau_n<T\}\right]\\
  &= -n\E^\Qmin[\log\sE(L)_T;\{\sE(L)_T>n\}\cap\{\tau_n<T\}] \\
  &\ \,\,\,\  -n\E^\Qmin[\log\sE(L)_T;\{\sE(L)_T\leq n\}\cap\{\tau_n<T\}] \\
  &\leq 0 + nK\Qmin(\tau_n<T).
\end{align*}
Equation \eqref{eqn:to_show} now follows from
\begin{align*}
  0 
  &\leq n\log n\, \Qmin(\tau_n<T)-n\,\E^\Qmin[\log\sE(L)_T;\,\{\tau_n<T\}] \\
  &\leq n\log n\, \Qmin(\tau_n<T) + nK\,\Qmin(\tau_n<T)\\
  &\longrightarrow 0, \ \ \text{ as $n\rightarrow\infty$}.
\end{align*}
\end{proof}

\end{example}

\bibliographystyle{plain}
\bibliography{mkt_stability_bib}

\end{document}